\documentclass[12pt]{amsart}
\newtheorem{prop}{Proposition}[section]

\newtheorem{cor}{Corollary}

\title
[Normal forms and gauge symmetries of local dynamics] {Normal
forms and gauge symmetries of local dynamics}
\author{ S.L. Lyakhovich and  A.A. Sharapov}
\address{Department of Quantum Field Theory, Tomsk State University, Tomsk 634050, Russia}
\email{sll@phys.tsu.ru, sharapov@phys.tsu.ru}

\begin{document}

\begin{abstract}
A systematic procedure is proposed for deriving all the gauge
symmetries of the general, not necessarily variational, equations
of motion.  For the variational equations, this procedure reduces
to the Dirac-Bergmann algorithm for the constrained Hamiltonian
systems with certain extension: it remains applicable beyond the
scope of Dirac's conjecture. Even though no pairing exists between
the constraints and the gauge symmetry generators in general
non-variational dynamics, certain counterparts still can be
identified of the first- and second-class constraints without
appealing to any Poisson structure. It is shown that the general
local gauge dynamics can be equivalently reformulated in an
involutive normal form. The last form of dynamics always admits
the BRST embedding, which does not require the classical equations
to follow from any variational principle.
\end{abstract}
  \maketitle

\section{Introduction}
By a local dynamical system we understand the dynamics whose true
trajectories are defined by a finite system of ordinary
differential equations. Given a local dynamical system, the
question arises of finding all its gauge symmetries. If the
equations of motion were variational, the Dirac-Bergmann algorithm
\cite{Dirac}, \cite{HT} would do the job. In fact, the
Dirac-Bergmann algorithm does even more: it iteratively brings the
variational equations of motion to a canonical form of the
constrained Hamiltonian dynamics with a complete set of
constraints on the phase-space variables. In this form, the
equations of motion are totally self-contained, having no other
constraints that can be derived from any compatibility condition.
It is the normal form of the variational dynamics, which is
exploited in physics for many different purposes. Besides the
other things, this canonical form allows one to classify the
constraints by grouping them into the first and second classes,
and it also results in identifying the gauge symmetry generators
as the Hamiltonian vector fields associated to the first-class
constraints. For the general (i.e., not necessarily variational)
local dynamics, no canonical form has been yet identified  which
might be viewed as an equivalent of the Hamiltonian dynamics with
a complete set of constraints. The correspondence between the
complete set of first-class constraints, including the primary and
secondary ones, and the gauge symmetries is known as the Dirac
conjecture. This conjecture is not always valid even for the
regular variational equations, and the counterexamples are well
known \cite{HT}. Beyond the scope of validity of the Dirac
conjecture, no systematic procedure has been known for identifying
the gauge symmetries even for the constrained Hamiltonian
dynamics.

In this paper, we work out an algorithm of bringing  the general
(not necessarily variational) dynamics to a certain normal form.
In this form, the gauge  symmetries and the complete set of
constraints are explicitly identified. Since no natural pairing
exists between the constraints and gauge symmetries unless the
dynamics are Hamiltonian, our algorithm derives them
independently: the constraints are found first, and then the
symmetries are identified. Despite of the fact that Dirac's
classification of constraints is conventionally defined in terms
of the Poisson brackets, the notions of the first- and
second-class constraints can be naturally extended to the
non-Hamiltonian systems. The algorithm works equally well in the
constrained Hamiltonian dynamics, finding out all the gauge
symmetries, even though the Dirac conjecture does not hold for the
system. So, the method may have some new impact on the
well-studied area of the constrained Hamiltonian dynamics.

Below we give some introductory comments on the background of the
problem and outline the contents  of the paper.

Any system of ordinary differential equations can always be
depressed to the first order  by introducing new variables, so we
start with a first-order autonomous system. Under certain
regularity conditions\footnote{From the viewpoint of physics, the
regularity means that the system has a definite number degrees of
freedom, i.e., the equations admit the same number of independent
Cauchy data in every domain of the configuration space. A more
accurate formulation of the regularity conditions is given in
Section 2 as well as the explanation how to bring the general
regular equations to the normal form (\ref{lambda}),
(\ref{Constr1}).} the first order ODE system can be always brought
(by adding new auxiliary variables, if necessary) to the following
form
\begin{equation}\label{lambda}
    \dot{x}{}^i=V^i(x) + \lambda^\alpha Z^i_\alpha(x)\, , \qquad \alpha =1, \ldots , l , \quad i= 1, \ldots ,
    n\,,
\end{equation}
\begin{equation}\label{Constr1}
T_a(x)=0\, , \qquad a=1, \ldots , m \, ,
\end{equation}
where $x^i, \lambda^\alpha$ are the decision variables. Thus, in
the regular local dynamics, it is always possible to have only two
types of variables: (i) the phase space coordinates $x^i$ being
subject to the normal differential equations (\ref{lambda}) and
the constraints (\ref{Constr1}), and (ii) the variables
$\lambda^\alpha$ entering linearly and  without derivatives into
the right hand sides of the differential equations for $x$'s. We
call the set of equations (\ref{lambda}), (\ref{Constr1}) the
\textit{primary normal form} of local dynamics. Obviously, the
compatibility conditions between the differential equations
(\ref{lambda}) and the algebraic constraints (\ref{Constr1}) can
result in more algebraic equations, which are called the secondary
constraints. The differential consequences of equations
(\ref{lambda}), (\ref{Constr1}) are considered in Section 3.

Two particular classes of equations  (\ref{lambda}),
(\ref{Constr1}) were extensively studied  by two different
sciences: Dirac's analysis of the constrained Hamiltonian systems
and optimal control theory. The first one deals with the case
where the number of the constraints coincides with the number of
$\lambda$'s, i.e.,  $l=m$, the vector fields $Z_\alpha$ and $V$
are all Hamiltonian, and
\begin{equation}\label{Dirac}
    Z^i_a = \{ x^i \, , \, T_a \} \, , \qquad V^i = \{
    x^i \, , \, H \} \, .
\end{equation}
Here $\{\cdot \,, \cdot \}$ is a non-degenerate Poisson bracket
and $T_a$ are the algebraic constraints (\ref{Constr1}). Upon this
identification, the equations of motion become variational,
following from the least action principle for the functional
\begin{equation}\label{SDirac}
    S[x,\lambda] = \int \left( \rho_i(x) \dot{x}{}^i -H(x) -
    \lambda^aT_a(x)\right)dt \, ,
\end{equation}
where $\rho_i dx^i$ is a symplectic potential associated to  the
Poisson bracket. This can be understood as a conditional extremum
for the standard Hamiltonian action subject to the constraints
(\ref{Constr1}), with the variables $\lambda$ being the
corresponding Lagrange multipliers. It is the variational dynamics
that the Dirac-Bergmann algorithm is normally applied to. If,
after applying the algorithm, the theory appears containing the
first-class constraints, the solutions are not unique for the
Cauchy problem and the action (\ref{SDirac}) possesses gauge
symmetries.

Optimal control theory deals with another particular case of the
above equations, which is opposite, in a sense, to the variational
case: only the differential equations (\ref{lambda}) are
considered, with no constraints (\ref{Constr1}) imposed on $x$'s.
In optimal control \cite{CMP}, \cite{SA},  the $\lambda$'s are
called the control functions. If only  equations (\ref{lambda})
are considered, the Cauchy problem for $x$'s will have a unique
solution corresponding to every choice of the functions
$\lambda^\alpha (t)$. The control functions remain unrestricted
anyhow by equations (\ref{lambda}), then the solutions remain
ambiguous for $x^i$ unless all $\lambda$'s are determined by some
extra requirement. The distinctions between Dirac's constrained
dynamics and optimal control theory go far beyond imposing (or not
imposing) constraints and supposing (or not supposing) the
existence of a symplectic structure: they profess quite opposite
concepts of the ambiguities in solutions of equations
(\ref{lambda}). Optimal control theory considers the functions
$\lambda$ as describing a background for the dynamics, making it,
in fact, non-autonomous. The objective is to minimize certain
function(al)s of the solutions $x^i(t)$ (the cost functions) by
varying the control functions $\lambda (t)$. Quite opposite, the
Dirac constrained dynamics consider all the solutions for $x(t)$
and $\lambda (t)$ equivalent to each other whenever they
correspond to the same Cauchy data. In particular, a function(al)
is considered as physically observable, if it evolves in the same
way on every trajectory from the equivalence class, i.e., it is
out of control from the viewpoint of optimal control theory. The
gauge theory treats $\lambda$ as dynamical variables, not the
functions controlled from outside. From this viewpoint, the
controllable values are supposed to be unobservable because the
autonomous physical models must be self-contained, leaving no room
for intervention of any ultramundane force that can
control/optimize what we observe. With this regard, the ambiguity
brought to the solutions by the undetermined multipliers $\lambda
(t)$ is to be factored out from the dynamics, first classically
and then quantum-mechanically. A basic tool for such a
factorization is a gauge symmetry relating equivalent
trajectories. In particular, the physical observables are
identified with the gauge invariant function(al)s.

By a gauge symmetry of equations (\ref{lambda}), (\ref{Constr1})
we understand the infinitesimal transformations of the form
\begin{equation}\label{gt}
    \delta_\varepsilon x^i = \sum_{n=0}^p R_{_{(p-n)}}^i
    \stackrel{_{(n)}}{\varepsilon}{} \, , \qquad \delta_\varepsilon
    \lambda^\alpha =  \sum_{n=0}^{p+1}U_{_{(p+1-n)}}^\alpha
    \stackrel{_{(n)}}{\varepsilon}\, .
\end{equation}
The transformation parameters  $\varepsilon$ are arbitrary
functions of time, and $\stackrel{_{(n)}}{\varepsilon}{}$ stands
for the $n$-th order time derivative of $\varepsilon{}$. The
structure functions $R^i_{^{(n)}}$ and $U^\alpha_{^{(n)}}$ are
allowed to depend on a finite number of variables $x^i$,
$\lambda^\alpha$, $\dot\lambda {}^{\alpha},\ddot\lambda
{}^{\alpha},...$ . The transformation (\ref{gt}) is supposed to
leave equations (\ref{lambda}), (\ref{Constr1}) invariant in the
sense that the  variations
\begin{equation}\label{deltae}
\delta_\varepsilon \left( \dot{x}{}^i -V^i(x) -
Z^i_\alpha(x)\lambda^\alpha \right)\,,\qquad
\delta_{\varepsilon}T_a(x)
\end{equation}
must be  proportional to the equations of motion (\ref{lambda}),
(\ref{Constr1}) for arbitrary $\varepsilon (t)$. In what follows
we use the ordinary physical terminology and refer to the values
that are proportional to the equations of motion as vanishing
\textit{on-shell}.

Consider, for example, unconstrained equations (\ref{lambda}) and
suppose that the vector fields $Z_\alpha$ span an integrable
distribution of rank $m$,
\begin{equation}\label{inv}
    [Z_\alpha \, , \, Z_\beta ]= U_{\alpha\beta}^\gamma (x)
    Z_\gamma \, ,
\end{equation}
which is preserved by $V$,
\begin{equation}\label{Vinv}
 [Z_\alpha \, , \, V ]= V_\alpha^\beta (x) Z_\beta \, .
\end{equation}
Then there are exactly $m$ linearly independent gauge
transformations of the form (\ref{gt}). These read
\begin{equation}\label{gt0}
\delta_\varepsilon x^i = Z_\alpha^i \varepsilon^\alpha \, , \qquad
\delta_\varepsilon \lambda^\alpha = \dot{\varepsilon}{}^\alpha
-\left( V^\alpha_\beta + U_{\beta\gamma}^\alpha\lambda^\gamma
\right)\varepsilon^\beta \, .
\end{equation}
Notice that equations (\ref{lambda}) are fully consistent, no
mater what are the vector fields $V$ and $Z_\alpha$ involved in
the right hand side. The involution relations (\ref{inv}),
(\ref{Vinv}) are not the necessary conditions for consistency of
the equations. If the distribution
$\mathcal{Z}=\mathrm{span}\{Z_\alpha\}$ is not integrable and/or
not preserved by $V$, the explicit form of the gauge
transformation (\ref{gt}) has remained unknown yet, and we are
going to present it in Section 4.

Given the distinctions in the main concepts/objectives of the
gauge theory and optimal control, it is no surprise that the gauge
symmetry remained unstudied for the general dynamical equations
(\ref{lambda}), (\ref{Constr1}) in the context of optimal control.
Even though a lot of interesting non-variational gauge models
arise in the context of modern higher energy physics, e.g.
self-dual Yang-Mills fields, higher-spin gauge theories,
Sieberg-Witten and Donaldson-Uhlenbeck-Yau equations, etc, the
issue of extending the Dirac-Bergmann method to the general local
dynamics  has been shelved yet, because no ways were seen to
quantize the non-variational dynamics. Recently, we have proposed
the BRST quantization methods applicable to arbitrary
non-variational dynamics \cite{LS1} (see also \cite{CaFe} for
another version of this method). Our method implies the
characteristic distribution to be on-shell involutive and tangent
to the constraint surface. As it has been already explained, this
is not the most general case. In Section 5, we explain that any
local gauge dynamics can be eventually reformulated in the
involutive form, which is physically equivalent to the original
one. The involutive normal form of dynamics, resulting from the
extension of the Dirac-Bergman algorithm to non-variational
dynamics, thus becomes a starting point for applying the
deformation quantization techniques for the general dynamical
systems.

\section{Regularity conditions and the primary normal form of  local dynamics}

Let us start with explanation how to bring a local dynamical
system  to the normal form (\ref{lambda}), (\ref{Constr1}). Any
finite system of ordinary differential equations  can always be
depressed to the first order by introducing new variables. By
further adding new variables, if necessary, the first-order ODEs
can be equivalently rewritten in the form of inhomogeneous
Pfaffian equations
\begin{equation}\label{Pf1}
    \theta_{J i}(x)\dot{x}{}^i=V_J(x) \, , \qquad J=1, \ldots , N \,,\quad i= 1, \ldots ,
    n\,,
\end{equation}
subject to constraints
\begin{equation}\label{C1}
T_a(x)=0\, , \qquad a=1, \ldots , m \, .
\end{equation}
For example, given second-order equations
\begin{equation}\label{f}
    f(y, \dot{y}, \ddot y) = 0 \, ,
\end{equation}
by adding new variables $v$ and $w$, these can be equivalently
rewritten in the form (\ref{Pf1}), (\ref{C1}) as
\begin{equation}\label{depress}
   \dot v=w\,,\qquad \dot{y}=v \, , \qquad f(y,v,w)=0 \, .
\end{equation}
Sometimes, to take into account the geometric properties of the
original equations (\ref{f}), it might be more appropriate to
introduce the new variables $v$, absorbing the derivatives of $x$,
as linear inhomogeneous functions of $\dot x$, $v=e(x){\dot x}+
\Gamma(x)$. This will result, however, in the equations of the
form (\ref{Pf1}, \ref{C1}), anyway. We call the \textit{primary
constraints} or just \textit{constraints} both equations
(\ref{C1}) and the functions $T_a(x)$.

To further proceed with the general equations of motion in the
constrained Pfaffian form (\ref{Pf1}), (\ref{C1}), we need to
impose appropriate regularity conditions. The phase space of the
dynamical system is supposed to be an open domain $U$ in
$\mathbb{R}^n$ with linear coordinates $\{x^i\}$. All the
functions of $x^i$ entering equations (\ref{Pf1}), (\ref{C1}) are
supposed to be analytic. \footnote{The conditions we adopt here
are not the weakest possible ones. For example, one can extend the
consideration from $\mathbb{R}^n$ to an analytic or smooth
manifold $M$, with constraints $\{T_a(x)\}$ being a section of a
vector bundle over $M$. In this paper, however, we are going to
avoid the technicalities related to less restrictive regularity
conditions, for the sake of a more clear presentation of the
algorithm as such. }

Let us introduce some convenient notions for describing the
regularity of the equations. Let $M$ be a matrix whose  elements
are analytic functions on $\mathbb{R}^n$. Given a  subset
$U\subset \mathbb{R}^n$, one can define a sequence of embedded
subspaces
\begin{equation}\label{U}
    U=U^0\supset U_M^1\supset U_M^2\supset\cdots\supset
    U_M^m=U'_M\,,
\end{equation}
where
\begin{equation}\label{Um}
    U_M^r=\{x\in U|\;\mathrm{rank}M(x)\geq r\}
\end{equation}
is the  set of all points in which the rank of the matrix $M$ is
not less than $r$, and $m$ is the maximal value of $\mathrm{rank}
M$ on $U$. Then $U_M^r$ is an open everywhere dense subset of $U$.
A minimal subset $U'_M\subset U$ will be called a \textit{a
regular part of $U$ regarding $M$}. It will be denoted just by
$U'$ whenever the corresponding matrix-valued function $M$ is
obvious from the context. The complementary subset to the regular
part, $U\backslash U'$, will be called an \textit{abnormal part}
of $U$.

Denote by $\Sigma$ the common zero locus of the constraints
(\ref{C1}), i.e.,
\begin{equation}\label{Sigma0}
    \Sigma=\{x\in U| \; T_a(x)=0 \, ,
    a=1,...,m\}\,.
\end{equation}
The primary constraints are assumed to be consistent as algebraic
equations per se. If the constraints contradict each other, the
entire system of the dynamical equations (\ref{lambda}),
(\ref{Constr1}) does not have any solution, and this can  hardly
be considered as regular dynamics. The same consistency of the
constraints by themselves will always be assumed  for the
secondary constraints derived in the next section.

Let $U'$ be the regular part of $U$ regarding the Jacobi matrix of
the constraints $J=(\partial_iT_a)$. The constraints $\{T_a\}$ are
called \textit{regular} if
\begin{equation}\label{Siprime}
\Sigma'=U'\cap\Sigma\neq \emptyset.
\end{equation}
In this case, $\Sigma'$ is a smooth submanifold in $U'$ with
$\dim{\Sigma}' = n-\mathrm{rank}J|_{U'}$.

 Denote  by $\Sigma''$
the  regular part of $\Sigma'$ regarding the matrix
$\Theta=(\theta_{Ji})$ of the Pfaffian forms in (\ref{Pf1}).  In
accordance with the definitions above, $\Sigma''\subset
\widetilde{U}=(U')^r_{\Theta}$, where $r=\mathrm{rank}\Theta(p)$
for $p\in \Sigma''$.  Equations (\ref{Pf1}) are linear in
velocities $\dot x^i$, with the coefficient matrix $(\theta_{Ji})$
having rank greater than or equal to $r$ on $\widetilde{U}$.
Restricting the dynamics to $\widetilde{U}$, we can solve
(\ref{Pf1}) with respect to the velocities as
\begin{equation}\label{l1}
    \dot x^i=V^i(x)+\lambda^\alpha Z_\alpha^i(x)\,,\qquad
    \alpha=1,...,l,\qquad l=n-r\,.
\end{equation}
Here $\lambda^\alpha(t)$ are arbitrary functions of time, the
vector field $V$ is any solution to the inhomogeneous linear
equations
\begin{equation}\label{Vi}
 (\theta_{Ji} V^i- V_j )|_{\Sigma^{\prime\prime}}=0\,,
\end{equation}
and the vector fields $\{Z_\alpha\}$ span the space of solutions
to the corresponding linear system
\begin{equation}\label{Z}
    \theta_{Ji}Z^i|_{\Sigma^{\prime\prime}}=0 \,.
\end{equation}
Upon restriction to the domain $\widetilde{U}$, the system of
equations  (\ref{l1}) and (\ref{C1}) is completely equivalent to
the original equations (\ref{Pf1}) and (\ref{C1}). We refer to
equations (\ref{l1}), (\ref{C1}) as a \textit{primary normal form}
of local dynamics. It is the form (\ref{lambda}), (\ref{Constr1})
announced in Introduction. The vector field $V$ will be called a
\textit{primary drift}, and the vector distribution
$\mathcal{Z}=\mathrm{span}\{Z_\alpha\}$ will be called a \textit{
primary characteristic distribution}. Notice that the distribution
$\mathcal{Z}$ is not required to have a constant rank on
$\widetilde{U}$ nor is it supposed to be involutive.

The primary normal form of equations (\ref{lambda}),
(\ref{Constr1}) implies certain equivalence relations among their
ingredients. Besides  the non-degenerate changes of the decision
variables
\begin{equation}\label{change}
    x^i \mapsto x'{}^i=f^i(x)\,,\qquad \lambda^\alpha \mapsto
    \lambda'{}^\alpha=F^\alpha_\beta(x)\lambda^\beta+G^\alpha(x)\,,
\end{equation}
the system is invariant with respect to the following
redefinitions:
\begin{equation}\label{aut}
    T_a \mapsto T'_a=S_a^b T_b\,, \qquad V\mapsto V'=V +
    T_aW^a\,,\qquad Z_\alpha\mapsto Z'_\alpha=Z_\alpha +
    T_aX_\alpha^a\,,
\end{equation}
where $(S_a^b(x))$ is an arbitrary non-degenerate matrix on
$\widetilde{U}$, and $W^a$ and $X^a_\alpha$ are arbitrary vector
fields. So, neither the primary drift $V$, nor the characteristic
distribution $\mathcal{Z}$ is defined uniquely by the
inhomogeneous Pfaffian system (\ref{Pf1}) subject to the
constraints (\ref{C1}). We say that two  dynamical systems on
$\widetilde{U}$ are \textit{equivalent} if their primary normal
forms (\ref{lambda}), (\ref{Constr1}) are related by
transformations (\ref{change}), (\ref{aut}). In what follows, we
will deal with those dynamical systems which admit at least one
analytic representative in each equivalence class, and only the
analytic representations will be considered.

A function (or vector field) $F$ is called trivial if it vanishes
on $\Sigma^\prime $. With account of regularity we have
\begin{equation}\label{triv}
 F|_{\Sigma'}=0 \quad \Leftrightarrow\quad   F= T_aF^a
\end{equation}
for some smooth functions (or vector fields) $F^a$. From this
viewpoint, the automorphisms (\ref{change}), (\ref{aut}) imply
that two characteristic distributions are equivalent whenever
their difference is trivial, and two drifts are equivalent
whenever they coincide modulo the characteristic distribution and
a trivial vector field.

As it has been already mentioned in Introduction, the general
equations of motion, being brought to the primary normal form
(\ref{lambda}), (\ref{Constr1}), can have nontrivial compatibility
conditions resulting in additional constraints on $x$'s and
$\lambda$'s. The secondary constraints on $x$'s can be derived
directly in the inhomogeneous Pfaffian form (\ref{Pf1}),
(\ref{C1}), as it was shown by X. Gracia and J.M. Pons \cite{GP}.
In the next section, we consider the procedure of deriving the
secondary constraints and fixing the undetermined multipliers
$\lambda$ by making use of the primary normal form of the
equations of motion. This procedure allows us to extend the Dirac
classification of constraints to the general equations
(\ref{lambda}), (\ref{Constr1}), without appealing to any Poisson
structure.

\section{Consistency conditions and secondary constraints}
The consistency conditions of the primary normal equations
originate from the requirement that the integral trajectories of
the flow (\ref{lambda}), with fixed functions $\lambda (t)$,
should be confined at the regular part of the constraint surface
(\ref{Siprime}). For the regular constraints this means that the
time derivatives of all the constraints must vanish on the
integral trajectories whenever they intersect $\Sigma^\prime$:
\begin{equation}\label{dotT}
\dot{T}_a{}|_{\mathrm{on-shell}}=\left( VT_a+\lambda^\alpha
Z_\alpha T_a \right)|_{ \Sigma^\prime } = 0 \, .
\end{equation}
The expression in brackets being a trivial function in the sense
of (\ref{triv}), we arrive at the following linear inhomogeneous
equations for $\lambda$:
\begin{equation}\label{dT}
 VT_a+ \lambda^\alpha Z_\alpha T_a =F_a^bT_b \,.
\end{equation}
Let $\Sigma''$ be the regular part of $\Sigma'$ regarding the
matrix $M=(Z_\alpha T_b)$, and $\mathrm{rank}M(p)=s$ for $ p\in
\Sigma''$. Then, according to the definition (\ref{Um}),
$\Sigma''\subset (U')_M^s$.

To further proceed with solving equations (\ref{dT}), we decompose
the constraints, characteristic distribution  and undetermined
multipliers in the following way:
\begin{equation}\label{decomp}
 T_a= (T_A,T_{{\bar{a}_1}})\,,\qquad
 Z_\alpha=(Z_A,Z_{{\alpha}_1})\,,\qquad
 \lambda^\alpha=(\lambda^A,\lambda^{{\alpha}_1})\,,
\end{equation}
where $A=1,...,s$ and $D=(Z_AT_B)$ is a maximal non-degenerate
minor of $M$. Denote $\stackrel{_{(1)}}{U}$ the regular part of
$(U')_M^s$ regarding $D$. Using equations (\ref{dT}) with $a=A$,
we can express $\lambda^A$ as linear inhomogeneous functions of
$\lambda^{\alpha_1}$:
\begin{equation}\label{lA}
    \lambda^A= - D^{AB}(VT_B
    +\lambda^{{\alpha_1}}Z_{{\alpha_1}}T_B)\,,
\end{equation}
with the matrix $(D^{AB})$ being inverse to $D$. After
substitution of the determined multipliers (\ref{lA}) into the
remaining equations (\ref{dT}), we get
\begin{equation}\label{Tprime}
    {T}^1_{\bar{a}_1}+\lambda^{\alpha_1}
    B_{{\alpha_1}{\bar{a}_1}}=0\,,
\end{equation}
where
\begin{equation}\label{B}
{T}^1_{\bar{a}_1}=VT_{\bar{a}_1}
-D^{AB}(VT_B)Z_AT_{\bar{a}_1}\,,\qquad
B_{{\alpha_1}{\bar{a}_1}}=Z_{\alpha_1}T_{
\bar{a}_1}-D^{AB}(Z_{\alpha_1}T_B)Z_AT_{\bar{a}_1}\,.
\end{equation}
Obviously, the matrix $B=(B_{{\alpha_1}{\bar{a}_1}})$ must vanish
on $\Sigma^\prime$, since otherwise one could determine more
multipliers $\lambda$, that would contradict maximality of the
non-degenerate minor $D$. Being a trivial analytic function in
$\stackrel{_{(1)}}{U}$,  the matrix $B$ has the form
$B_{{\alpha_1}{\bar{a}_1}}=W_{{\alpha_1}{\bar{a}_1}}^cT_c$.
Equations (\ref{Tprime}) should be satisfied together with the
equations of primary constraints, hence the trivial terms
containing $\lambda$'s can be omitted and we obtain the new
constraints
\begin{equation}\label{C2}
{T}^1_{\bar{a}_1} (x) = 0 \, .
\end{equation}
So, we see that the conservation laws (\ref{dotT}) for the primary
constraints are equivalent to equations (\ref{lA}) determining a
part of $\lambda$'s as specific functions of $x$'s, and the
$\lambda$-independent relations (\ref{C2}). Equations (\ref{C2})
as well as the functions ${T}^1_{\bar{a}_1}$ themselves  are
called the \textit{secondary constraints of the first stage}.
Notice that the secondary constraints are defined modulo trivial
contributions proportional to the primary constraints. Also notice
that the functions ${T}^1(x)$ are not necessarily independent.

Substituting the fixed multipliers (\ref{lA}) into the original
equations of motion  (\ref{lambda}), (\ref{Constr1}) and adding
the secondary constraints,  we arrive at the following set of
equations:
\begin{equation}\label{secondary}
    \dot x^i= \stackrel{_{(1)}}{V}{}^i(x)+\lambda^{{\alpha_1}}
    \stackrel{_{(1)}}{Z}{}_{{\alpha}_1}^i(x)\,,\qquad \stackrel{_{(1)}}{T}_{a_1}(x)=0\,,
\end{equation}
where the values
\begin{equation}\label{VZT1}
    \stackrel{_{(1)}}{V}=V - D^{AB}(VT_B)Z_A\,,\qquad
    \stackrel{_{(1)}}{Z}_{{\alpha}_1}=Z_{
    \alpha_1}-D^{AB}(Z_{{\alpha}_1}T_B)Z_A\,,\qquad
    \stackrel{_{(1)}}{T}_{a_1}=(T_a, T_{\bar{a}_1}^1)\,.
\end{equation}
are called, correspondingly, the drift, characteristic
distribution, and constraints of the first stage. By construction,
all these objects are well defined on the open everywhere dense
domain $\stackrel{_{(1)}}{U}\subset U$.

Let us write down  the  following obvious relations:
\begin{equation}\label{stage1}
\stackrel{_{(1)}}{V}T_A=0\,,   \qquad
\stackrel{_{(1)}}{V}T_{{\bar{a}_1}} = {T}^1_{{\bar{a}_1}}\,,\qquad
\stackrel{_{(1)}}{Z}_{\alpha_1}T_A=0\,, \qquad
\stackrel{_{(1)}}{Z}_{\alpha_1}T_{a_1}=W_{\alpha_1 a_1}^c T_c \, .
\end{equation}

Along with the condition of invertibility of the matrix
$D=(Z_AT_B)$ these relations give a tip for a simple geometric
interpretation of the transition from the original equations
(\ref{lambda}), (\ref{Constr1}) to the equivalent first-stage
equations (\ref{secondary}). The original surface of primary
constraints (\ref{Sigma0}) can be represented (at least locally)
as a transverse intersection $\Sigma= \Sigma^{^{\|}}\cap
\Sigma^{^{\bot}}$ of two surfaces $\Sigma^{^{\|}}$ and
$\Sigma^{^{\bot}}$. These are given by
\begin{equation}\label{decompS}
    \Sigma^{^{\|}}=\{x\in\, \stackrel{_{(1)}}{U}|\; T_{a_1}(x)=0\,\}\,,\qquad
    \Sigma^{^{\bot}}=\{x\in\, \stackrel{_{(1)}}{U}|\; T_A(x)=0\,\}\,.
\end{equation}
In its turn, the original characteristic distribution
$\mathcal{Z}=\mathrm{span}\{Z_\alpha\}$ is decomposed into the sum
 $\mathcal{Z}=\mathcal{Z}^{^{\|}}\oplus
\mathcal{Z}^{^{\bot}}$ of two sub-distributions
\begin{equation}\label{Z|bot}
\mathcal{Z}^{^{\|}}=\mathrm{span}\{\stackrel{_{(1)}}{Z}_{\alpha_1}\}\,,\qquad
\mathcal{Z}^{^{\bot}}=\mathrm{span}\{Z_A\}\, ,
\end{equation}
which are called, correspondingly, the tangential and transverse.
The multipliers $\lambda^A$, being related to the transverse
sub-distribution, are fixed by the conservation law of the primary
constraints, whereas the multipliers $\lambda^{\alpha_1}$
corresponding to $\mathcal{Z}^{^\|}$ still remain undetermined.

Geometrically, relations (\ref{stage1}) mean that the distribution
$\mathcal{Z}^{^{\|}}$ is tangent to the regular part of the
primary constraint surface $\Sigma$, whereas the complementary
distribution $\mathcal{Z}^{^{\bot}}$ is transverse. The tangential
distribution $\mathcal{Z}^{^{\|}}\subset \mathcal{Z}$ is
invariantly defined by the property
$\mathcal{Z}^{^{\|}}|_{\Sigma}\subset T\Sigma$. The dimension of
the transverse distribution, being complementary to the dimension
of $\mathcal{Z}^{^{\|}}$, coincides with the number of the
multipliers (\ref{lA}) determined from the conservation
requirement for the primary constraints. The primary
characteristic distribution $\mathcal{Z}$ is tangent to the
surface $\Sigma^{^{\|}}$, being zero locus of the primary
constraints not involved in determining of the multipliers. Also
notice that the first-stage drift $\stackrel{_{(1)}}{V}$ is
tangent to the primary constraint surface $\Sigma$ on the zero
locus of the secondary constraints.

Not only do the first-stage equations (\ref{secondary})  describe
the same dynamics as (\ref{lambda}), (\ref{Constr1}) (the
zero-stage equations), but they appear identical to equations
(\ref{lambda}), (\ref{Constr1}) in form. Therefore, we can apply
to (\ref{secondary}) all the above reasonings concerning the
preservation of the constraints in time. Supposing the first-stage
constraints to be regular, we check whether the conservation
condition
\begin{equation}\label{}
\begin{array}{c}
\frac{d}{dt}\stackrel{_{(1)}}{T}_{a_1}\big|_{\mathrm{on-shell}}=0
\end{array}
\end{equation}
implies new secondary constraints, or it further restricts the
undetermined multipliers $\lambda$. The new constraints, if any,
are called the secondary constraints of the second stage and
denoted by ${T}^2_{\bar{a}_2}$. We suppose that the complete set
of the second-stage constraints $\stackrel{_{(2)}}{T}_{a_2}=(T_a,
T^1_{\bar a_1}, T^2_{\bar a_2})$ is regular.  After exclusion of
the determined  multipliers from the differential equations, we
 get  the drift and the characteristic distribution of the second
 stage. If there are nontrivial secondary constraints among the constraints of the second stage
they must conserve. So, we further proceed with deriving
consequences from the conservation of the second-stage
constraints. The algorithm will continue working, defining some
multipliers and redistributing the corresponding vector fields
from the tangential distribution of the previous stage to the
transverse distribution of the next stage, redefining the drift,
and bringing the constraints of the next stage. Since the overall
number of the independent regular constraints at any stage cannot
exceed the dimension of the phase space, the algorithm has to
stabilize after a finite number of steps. The stabilization is
achieved as soon as new secondary constraints stop to appear from
the consistency conditions. If the iterative procedure terminates
at the $k$-th stage, the equations of motion take the following
form:
\begin{equation}\label{im}
    \dot x^i= \stackrel{_{(k)}}{V}{}^i(x)+\lambda^{{\alpha_k}}
    \stackrel{_{(k)}}{Z}{}_{{\alpha}_k}^i(x)\,,\qquad \stackrel{_{(k)}}{T}_{a_k}=0\,.
\end{equation}
These equations are defined in an open everywhere dense domain
$\stackrel{_{(k)}}{U}\subset U$, and the complete set of
constraints
\begin{equation}\label{barS}
{\stackrel{_{(k)}}{T}_{a_k}} =\left( T_a, {T}^1_{\bar{a}_1},
\ldots , T^k_{\bar{a}_k} \right)
\end{equation}
includes the primary and secondary constraints of all stages.

By construction, the following relations take place:
\begin{equation}\label{VZim}
\stackrel{_{(k)}}{V} \stackrel{_{(k)}}{T}_{a_k} = F_{a_k}^{
b_k}\stackrel{_{(k)}}{T}_{b_k}\,, \qquad
\stackrel{_{(k)}}{Z}_{\alpha_k} \stackrel{_{(k)}}{T}_{a_k} =
F_{\alpha_k {a_k}}^{b_k} \stackrel{_{(k)}}{T}_{b_k}
\end{equation}
for some functions $F_{a_k}^{b_k}$  and $F_{\alpha_k
{a_k}}^{b_k}$. We call  (\ref{im}) a \textit{complete normal form}
of a local dynamical system. It is the form that ensures full
consistency of the dynamics as it does not have any further
compatibility conditions and consequences.

Several remarks can be made about the algorithm above:
\begin{itemize}
    \item The $(n+1)$-st step of the algorithm becomes possible whenever
    the constraints of $n$-th stage are regular. Actually, it is the constraint surface that
     has an invariant geometric meaning, not the constraints as functions.
    The same constraint surface  can be defined by different $T$'s, and
    all such constraints are considered to be equivalent. If the
    equivalence class includes a regular representative,
    this representative is picked up as the set of constraints of the $n$-th stage.
    As is seen from the iterative procedure, the algorithm is not sensitive
    to the specific choice of regular representative.

    \item The $(n+1)$-st iteration involves the ambiguity concerning the
    choice of a maximal non-degenerate minor $D_n$ of the matrix
    $M_n=(Z_{\alpha_n}\stackrel{_{(n)}}{T}_{a_n})$. The different
    choices can result in different domains of definition  $\stackrel{_{(n+1)}}{U}$
    for the system of the $(n+1)$-st stage. Notice that the different domains $\stackrel{_{(n+1)}}{U}$'s,
  being open everywhere dense subsets in $U$, should coincide modulo a set of measure zero.
  In the intersection of these almost coinciding domains of definition, the
  equations of motion are not sensitive to a specific choice of
  a maximal minor.

\item When the compatibility conditions (\ref{dT}) are considered
as a system of linear equations for $\lambda$'s, the surface
$\Sigma''$ was defined as  a subset of $\Sigma'$, where
$\mathrm{rank}\,M=s$ is maximal.
   This requirement of maximality can be relaxed, in principle. One can consider
   a subset $S\subset \Sigma'$, where $\mathrm{rank}\,M=s'$ is less than $s$.
   According to the terminology of Section 2, $S$ belongs to the abnormal part of the
   constraint surface $\Sigma'$ regarding $M$. In some cases, the abnormal
   set $S\subset U'$ can be a smooth submanifold defined by a set of regular constraints $T^S=0$.
   At every given stage, the dimension of the abnormal part of the constraint
surface is less than the dimension of the regular part, so that
$\dim S<\dim \Sigma$. Being restricted to $S$, the dynamics would
be regular in the same sense as with the restriction to the
regular part of the constraint surface. With the restriction of
the dynamics to the abnormal set $S$, a lesser number $s'$ of the
multipliers $\lambda$ are determined, but the greater number of
the constraints appear at this stage. Thus, at any stage of the
iterative procedure, there may be an ambiguity in deciding between
the regular or abnormal parts of the constraint surface. The
dynamics can be regular and consistent for either of these
options. It is amply clear that the dimension of the complete
constraint surface (\ref{barS}) and the number of undetermined
multipliers in (\ref{im}) can depend on the choice made at certain
stage (the algorithm ``bifurcates''). So, after applying the
algorithm, which works equally well with regular and abnormal
parts of the constraint surface, one can learn that the original
dynamical system (\ref{lambda}), (\ref{Constr1}) has contained in
itself several different dynamics (\ref{im}), with different
constraints, drifts and numbers of undetermined multipliers. These
dynamics never communicate and should be interpreted as different
physical systems, even though the difference reveals itself only
after applying the algorithm, being not explicitly visible in the
primary normal form of the dynamics.
\end{itemize}

Let us comment on the geometry of the dynamical equations
 in the complete normal form (\ref{im}).  The primary
characteristic distribution $\mathcal{Z}=\mathrm{span}\{Z_a\}$ is
eventually decomposed into the direct sum
\begin{equation}\label{Zdecomp}
    \mathcal{Z}=\mathcal{Z}_{^{\bot}}\oplus\mathcal{Z}_{^{\|}}
\end{equation}
of the transverse and tangential distributions with respect to the
complete constraint surface
\begin{equation}\label{comp}
    \bar\Sigma = \{x\in \stackrel{_{(k)}}{U}\;|\;
    \stackrel{_{(k)}}{T}_{a_k}=0\}\,.
\end{equation}
As is seen from (\ref{VZim}), the tangential distribution
$\mathcal{Z}_{^{\|}}=\mathrm{span}\{\stackrel{_{(k)}}{Z}_{\alpha_k}\}$
preserves the complete constraints surface  (\ref{comp}), and the
multipliers $\lambda^{\alpha_k}$ related to the basis vector
fields of $\mathcal{Z}_{^{\|}}$ remain arbitrary functions of time
in equations (\ref{im}). The distribution
$\mathcal{Z}_{^\bot}\simeq \mathcal{Z}/\mathcal{Z}_{^\|}$, being a
complement to $\mathcal{Z}_{^\|}$, is transverse to the complete
constraint surface ${\bar\Sigma}$: If $ \{ Z_{{\bar A}}\}$ is a
basis in $\mathcal{Z}_{^\bot}$, then
\begin{equation}\label{Zbot}
\mathrm{rank} (Z_{ {\bar A}} \stackrel{_{(k)}}{T}_{{a}_k}) = \dim
\mathcal{Z}_{^\bot} \,.
\end{equation}
 All the multipliers $\lambda^{\bar A}$ corresponding to
$\mathcal{Z}_{^\bot}$ are defined by the conservation conditions
of the primary or secondary constraints, at one or another stage
of the algorithm.

By making use of the equivalence transformations (\ref{aut}), the
complete set of constraints can be rearranged into the union of
two subsets
\begin{equation}\label{completeT|}
\stackrel{_{(k)}}{T}_{a_k}=( T^{_\|}_{\bar a} \, , \,
{T}^{^\bot}_{\bar A} )
\end{equation}
such that
\begin{equation}\label{}
\quad Z\, \bar{T}^{_\|}_{\bar a}|_{\bar{\Sigma}}=0  \quad \forall
Z\in \mathcal{Z} \quad \mbox{and}\quad
    \det D|_{\bar{\Sigma}}\neq
    0\, ,
    \end{equation}
where $D =(Z_{\bar{A}}{T}^{^\bot}_{\bar{B}})$. Denoting by
$\bar\Sigma_{^\|}$ and $\bar \Sigma_{^\bot}$ the zero loci of the
constraints $T^{^\|}$ and $T^{^\bot}$, respectively, we see that
the complete constraint surface  can be represented as the
intersection of two surfaces:
\begin{equation}\label{barSigmabot|}
\bar{\Sigma}=\bar{\Sigma}_{^\|}\cap\bar{\Sigma}_{^\bot} \, .
\end{equation}
The primary characteristic distribution $\mathcal{Z}$ is tangent
to $\bar{\Sigma}_{^\|}$ on $\bar{\Sigma}$ and transverse to
$\bar{\Sigma}_{^\bot}$.  Notice that
$\mathrm{codim}\bar{\Sigma}_{^\bot} =\mathrm{dim}
\mathcal{Z}_{^\bot}$, while the dimension of the tangential
distribution $\mathcal{Z}_{^\|}$ does not correlate with the
dimension of $\bar{\Sigma}_{^\|}$ in general.

The complete drift  $\bar V=\stackrel{_{(k)}}{V} $ defined by
(\ref{im}) is constructed step by step from the primary drift $V$
by adding terms proportional to the elements of the transverse
distribution. When the algorithm has terminated at the $k$-th
stage, the drift reads
\begin{equation}\label{barV}
    {\bar V} = V -
    (V{T}^{^\bot}_{{\bar A}})(D^{-1})^{{\bar A} {\bar B}}{Z}_{{\bar B}}
    \, .
\end{equation}
Notice that $ {\bar V}$ belongs to the equivalence class of the
primary drift $V$ in the sense of the equivalence relations
(\ref{aut}).

Let us elaborate on the correspondence between the basic notions
of the constrained Hamiltonian mechanics and their counterparts in
the  general local dynamics. The equations (\ref{im}) can be
viewed as an equivalent of the Dirac constrained dynamics with the
complete set of primary and secondary constraints (\ref{barS}).
The ``tangential'' and ``transverse'' constraints ${T}^{^\|}$ and
$T^{^\bot}$   correspond, respectively, to the {\it complete} sets
of the first- and second-class constraints in Dirac's
classification of the Hamiltonian constraints. The tangential and
transverse primary characteristic distributions
$\mathcal{Z}_{^\|}$ and $\mathcal{Z}_{^\bot}$ correspond to the
distributions generated by the Hamiltonian vector fields of the
{\it primary} first- and second-class constraints, respectively.
One may also wonder about an analog for the Hamiltonian vector
fields associated with the first-class constraints. These are
known to generate the whole set of gauge symmetries provided that
the Hamiltonian system obeys the Dirac conjecture. From the next
section we will learn that the vector fields in question form the
distribution
\begin{equation}\label{ZViter}
{\mathcal{Z}}_V = \mathcal{Z}_{^\|}\cup [ \mathcal{Z}_{^\|},
\mathcal{Z}_{^\|} ]\cup [\mathcal{Z}_{^\|}, {\bar V}] \cup
\cdots\,,
\end{equation}
where the dots stand for the higher iterated commutators of
$\mathcal{Z}_{^\|}$ and $\bar V$. Whether  the Dirac conjecture is
true or not, it is the distribution ${\mathcal{Z}}_V$ that
generates the gauge symmetry transformations (\ref{gt}) of the
dynamics.

\section{Gauge symmetries}

In the previous section, we have seen that the algorithm of
stabilization of the primary constraints brings the original
equations of motion to the complete normal form (\ref{im}),
(\ref{VZim}). In this section, we find all the gauge symmetry
transformations for these equations.

To ease the notation, we omit all the sub- and superscripts
referring to the final stage of the iterative procedure from the
previous section. After this omission, equations (\ref{im}) read
\begin{equation}\label{pr}
    \dot x^i=V^i(x)+\lambda^\alpha Z_\alpha^i(x)\,, \qquad T_a(x)=0\,,
\end{equation}
where the distribution $\mathcal{E}=\mathrm{span}\{V,Z_\alpha\}$
is assumed to be tangent to the complete constraint surface
$\bar{\Sigma}=\{x\in U|\; T_a(x)=0\}$.

 The infinitesimal gauge symmetry transformations
(\ref{gt}) are sought for equations (\ref{pr}) in the form:
\begin{equation}\label{gtr}
 \delta_\varepsilon
x^i=\sum_{n=0}^{p}R^i_{p-n}\stackrel{_{(n)}}{\varepsilon}\,,\qquad
\delta_\varepsilon
\lambda^\alpha=\sum_{n=0}^{p+1}U^\alpha_{p+1-n}\stackrel{_{(n)}}{\varepsilon}\,,
\end{equation}
where $R$'s and $U$'s are some functions of $x$, $\dot \lambda$,
$\ddot \lambda$, ... up to some finite order
$\stackrel{_{(k)}}{\lambda}$, and $\varepsilon (t)$ are the
transformation parameters, being arbitrary functions of $t$.
Notice that the integer $p$, called the \textit{order of gauge
transformation}, is not fixed a priori. Moreover, it can always be
risen by the trivial replacement $\varepsilon \mapsto \dot
\varepsilon$ of the gauge transformation parameter. Conversely, if
$R_p=0$ and $U_{p+1}=0$, then the reverse change of variables
$\dot \varepsilon \mapsto \varepsilon$ depresses the highest order
of the derivative of $\varepsilon$ in (\ref{gtr}). When the
functions $R_p$ and $U_{p+1}$ do not vanish
simultaneously\footnote{Below, we will prove the following fact:
For an undepressible gauge transformation all the vector fields
$R_{k}$, $k=0,1,...,p$, are linear independent. This means that
the gauge variation of $x$'s contains actually all the derivatives
$\varepsilon, \dot\varepsilon,
...,\stackrel{_{(p)}}{\varepsilon}$, without gaps.}, the highest
order of the  derivative cannot be depressed and we refer to
corresponding transformation as \textit{undepressible}.

Let us first consider the issue of gauge symmetry for the
dynamical system (\ref{pr}) without constraints $T_a$. The gauge
invariance of the dynamics with respect to the infinitesimal
transformations (\ref{gtr}) means that the gauge variation of
(\ref{pr}) should vanish on their solutions, i.e.,
\begin{equation}\label{dee}
    \delta_\varepsilon (\dot x^i-V^i-\lambda^\alpha
    Z^i_\alpha)|_{\mathrm{on-shell}}=0\,.
\end{equation}
Excluding the velocity  $\dot x^i$ in (\ref{dee}) with the help of
(\ref{pr}), we arrive at the following recurrent relations for the
structure functions $R$ and $U$ of the gauge transformation
(\ref{gtr}):
 \begin{equation}\label{rec1}
R^i_0=W^i_0\,,
\end{equation}
\begin{equation}\label{rec2}
 R^i_{n+1}=DR^i_n+W^i_{n+1}\,,\qquad
n=0,...,p-1\,,
\end{equation}
\begin{equation}\label{rec3}
 R^i_{p+1}\equiv DR_p^i+W^i_{p+1}=0\,.
\end{equation}
Here we have introduced the following abbreviations:
\begin{equation}
W^i_n=U_n^\alpha Z^i_\alpha\,,\qquad D=-\partial-[V+\lambda^\alpha
Z_\alpha,\,\cdot\;] \, ,
\end{equation}
and $\partial$  is understood as the time derivative acting only
on $\lambda$'s, i.e.,
\begin{equation}\label{dpart}
    \partial =\dot \lambda^\alpha\frac{\partial}{\partial
    \lambda^\alpha}+\ddot \lambda^\alpha\frac{\partial}{\partial\dot
    \lambda^\alpha}+\cdots\,.
\end{equation}
Equation (\ref{rec1}) implies that the coefficient at the highest
derivative of $\varepsilon$ in $\delta_{\varepsilon}x^i$ is given
by a linear combination of vector fields  from the characteristic
distribution $\mathcal{Z}$. Relations (\ref{rec2}) enable one to
express  all the structure functions $R$ in terms of the functions
$U$ defining the gauge transformations of the undetermined
multipliers $\lambda$. Namely, by solving relations (\ref{rec2}),
we find
\begin{equation}
R^i_n=\sum_{m=0}^nD^mW^i_{n-m}\,.
\end{equation}
In particular, equation (\ref{rec3}) takes the form
\begin{equation}\label{compcond}
R^i_{p+1}\equiv \sum_{m=0}^{p+1}D^mW^i_{p+1-m}=0\,.
\end{equation}
The last equation on $U$'s is the only nontrivial condition to
satisfy.  Its solutions form a linear space, which dimension is to
be computed, depending on $p$. To this end, certain regularity
conditions must be imposed on the equations of motion and their
domain of definition. Also, we will use some properties of and
notions about the distributions, which are briefly listed below.

So far, the components of the vector fields $V$, $Z_\alpha$  were
supposed to be analytic functions defined in some domain of linear
space. In the following, to avoid specifying each time the
definitional domain for the gauge transformations, we relax the
condition of analiticity and utilize the algebraic setting usually
adopted in control theory \cite{CMP} for the dynamical equations
(\ref{lambda}). Namely, $V$ and $Z_\alpha$ are allowed to be
meromorphic vector fields. As the scar set of a meromorphic
function is of measure zero,  our gauge transformations will be
well defined in an open everywhere dense domain $\widetilde{U}$ of
the original phase space. It is $\widetilde{U}$ that is to be
considered as the definitional domain of the gauge dynamics.

Let $F$ be the field of meromorphic functions and $W$ be the space
of meromorphic vector fields on $U\subset \mathbb{R}^n$. The space
$W$ has both the structure  of a real, infinite-dimensional Lie
algebra with respect to the commutator of vector fields and the
structure of $n$-dimensional vector space over $F$. We will refer
to $F$-linear subspaces of $W$ as \textit{distributions}.
Generally, an arbitrary distribution $P\subset W$ is not closed
with respect to the commutator, i.e., the commutator of two vector
fields from $P$ may not belong to $P$. If $[P,P]\subset P$, then
the distribution $P$ is said to be \textit{involutive}. The
\textit{Lie closure of $P$} is defined as the minimal involutive
distribution $\bar{P}\subset W$ that contains $P$. Uniqueness of
$\bar{P}$ follows immediately from the fact that the set of the
involutive distributions is closed  with respect to intersecting
of linear spaces. Furthermore, one can inductively see that the
distribution
 $\bar{P}$ is generated by the vector fields
\begin{equation}\label{itercom}
[\cdots[u_1,u_2],u_3],\cdots], u_k]\qquad \forall u_k\in P\,,\quad
\forall k\in \mathbb{N}\,.
\end{equation}
For an involutive distribution $P$,  $\bar{P}=P$.

By taking the multiple  commutators (\ref{itercom}), one can
filter the Lie closure $\bar{P}$ by the sequence of
sub-distributions
\begin{equation}\label{Dl}
\bar{{P}}_l=\mathrm{span}_\mathcal{F}\big\{\quad
[\cdots[u_1,u_2],u_3],\cdots], u_k]\quad \big |\quad \forall
u_i\in P\,,\quad k=1,...,l\quad \big\}
\end{equation}
such that
\begin{equation}\label{chainclosure}
\bar{P}= \bar{P}_m\supset \bar{P}_{m-1}\supset \bar{P}_{m-2}
\supset \cdots \supset \bar{P}_1=P \,.
\end{equation}
This filtration  is known as the \textit{Lie flag of the
distribution} $P$. Clearly, the minimal integer $m\leq n-\dim P$
involved in (\ref{chainclosure}) is an invariant of the
distribution $P$ together with the numbers $d_k=\dim \bar{P}_k$.
The integer $m$ is usually referred to as the \textit{depth of the
distribution} $P$ and the sequence of integers $(d_1,d_2,...,d_m)$
is called the \textit{growth vector} of  $P$.

To check the solvability of  equations (\ref{compcond}) we  need
some extension of the field $F$.  Along with the coordinates $x^i$
in $\mathbb{R}^n$ we introduce the infinite set of variables
$\{\lambda^\alpha_k\}$, $\alpha=1,...,m$, $k\in \mathbb{N}$.
Denote by $\mathcal{F}$ the field of meromorphic functions of
$x^i$ and a \textit{finite} number of the variables
$\lambda^\alpha_k$. Replacing in the definition of $W$ the field
$F$ by its extension $\mathcal{F}$, we get the $n$-dimensional
vector space
\begin{equation}\label{}
    \mathcal{W}=\mathrm{span}_{\mathcal{F}}\left(\frac{\partial}{\partial x^1}\,,\cdots,\frac{\partial}{\partial
    x^n}\right)\,.
\end{equation}
The $\mathcal{F}$-linear subspaces of the vector space
$\mathcal{W}$ will be called $\lambda$-distributions.

Let us turn $\mathcal{F}$ into a \textit{differential field} by
setting
$$
\partial x^i=0\,,\qquad \partial
\lambda^\alpha_k=\lambda^\alpha_{k+1}\,.
$$
Clearly, this definition just mimics the definition of the time
derivative (\ref{dpart}) if one set
$\lambda^\alpha_0=\lambda^\alpha$. The action of the operator
$\partial$ can be further extended to the $\lambda$-distributions
by the rule
$$
\partial(aw)=(\partial a)w+a\partial w\qquad \forall a\in
\mathcal{F}\,,\quad \forall w\in \mathcal{W}\,,
$$
$$
\partial \left(\frac{\partial}{\partial x^i}\right)=0\,.
$$

With all the definitions above, we can start studying the gauge
symmetries (\ref{gtr}) of the system (\ref{pr}). Every dynamical
system (\ref{pr}) defines and is defined by the distribution
$\mathcal{E}$  generated by the vector fields $V$ and $Z_\alpha$.
Associated to  $\mathcal{E}$ is a \textit{gauge distribution}
$\mathcal{Z}_V$.   The latter is defined as the limit of a
filtration
\begin{equation}\label{filter}
     \mathcal{Z}^0_V\subset \mathcal{Z}^1_V\subset \cdots
    \subset \mathcal{Z}^\infty_V=\mathcal{Z}_V\subset \mathcal{W}
\end{equation}
of involutive $\lambda$-distributions $\mathcal{Z}^k_V$ given by
\begin{equation}\label{}
    \mathcal{Z}_V^k=\overline{
    \bigcup_{m=0}^k
\underbrace{[V,...,[V}_{m},{\mathcal{Z}}]}\,.
\end{equation}
It is clear that $\mathcal{Z}_V^0=\bar{\mathcal{Z}}$ and
$\bar{\mathcal{E}}=\mathcal{Z}_V\cup V$.

Define an $\mathbb{R}$-linear operator $D:
{\mathcal{W}}\rightarrow {\mathcal{W}}$ by
\begin{equation}\label{D}
 Dw=-\partial w
-[V,w]-\lambda^\alpha[Z_\alpha, w ]\qquad \forall w\in
{\mathcal{W}}\,.
\end{equation}
Although the operator $D$ is not $\mathcal{F}$-linear, it
satisfies the following analog of the Leibnitz rule:
\begin{equation}\label{LR}
    D(aw)=(Da)w+aDw \qquad \forall a\in \mathcal{F}\,,\quad
    \forall w\in \mathcal{W}\,,
\end{equation}
where
\begin{equation}\label{}
    Da=-\partial a-Va-\lambda^\alpha Z_\alpha a\,.
\end{equation}
Since $D$ is a differentiation of the field $\mathcal{F}$, one can
thought of $\mathcal{W}$ as a differential $\mathcal{F}$-module.

Given an element $w\in \mathcal{Z}$, consider the sequence of
elements  $D^kw\in \mathcal{Z}_V$, $k\in \mathbb{N}$. As
${\mathcal{Z}_V}$ is of finite dimension, there is  $p \in
\mathbb{N}$ and a set of functions $a_1,\cdots, a_p\in
\mathcal{F}$ such that
\begin{equation}\label{1}
D^pw=a_1D^{p-1}w+a_2D^{p-2}w+\cdots +a_pw\,.
\end{equation}

\begin{prop}\label{Prop1} Let $w\in \mathcal{Z}$ satisfy equation (\ref{1}),
then there exist a sequence of elements $w_1,...,w_p\in
\mathcal{Z}$ such that
\begin{equation}\label{2}
D^pw+D^{p-1}w_1+D^{p-2}w_2+\cdots + w_p=0
\end{equation}
and $w_k=b_k w$ for some $b_k\in \mathcal{F}$.
\end{prop}

The proposition is proved by induction, making use of the Leibnitz
rule (\ref{LR}).

\begin{cor}  For any vector field $W_0$ from the characteristic
distribution $\mathcal{Z}=\mathrm{span}\{Z_\alpha\}$,  equation
(\ref{compcond}) has a solution for some $p$. In other words,
every basis vector field $Z_\alpha$ generates a gauge
transformation, so that the total number of independent gauge
parameters $\varepsilon$ coincides with $\dim \mathcal{Z}$.
\end{cor}

An element  $w\in \mathcal{Z}$ is said to have degree not higher
than $p$, if it satisfies equation (\ref{2}) with some (not
necessarily linear independent) $w_k\in \mathcal{Z}$. It is clear
that the elements of degree not higher than $p$ form a
$\lambda$-distribution $\mathcal{Z}_p\subset\mathcal{W}$. We have
the finite filtration
$$
\mathcal{Z}=\mathcal{Z}_{N}\supset
\mathcal{Z}_{N-1}\supset\cdots\supset \mathcal{Z}_1\supset 0\,.
$$
The numbers
$$
\delta_p=\dim \mathcal{Z}_p-\dim \mathcal{Z}_{p-1}
$$
are called the \textit{indices} of the characteristic distribution
$\mathcal{Z}$.

By making use of the Euclidean metric in $\mathbb{R}^n$, we can
split the imbedding $\mathcal{Z}_{p-1}\subset \mathcal{Z}_p$ as
$\mathcal{Z}_p=\mathcal{Z}_{p-1}\oplus \mathcal{Z}^\perp_{p-1}$.
As a result the characteristic distribution is decomposed into the
direct sum
\begin{equation}\label{decomp}
\mathcal{Z}=\mathcal{Z}^{1}\oplus
\mathcal{Z}^{2}\oplus\cdots\oplus \mathcal{Z}^{N}\,,\qquad
Z^p\simeq \mathcal{Z}_p/\mathcal{Z}_{p-1}\,,\qquad
\dim\mathcal{Z}^p=\delta_p\,.
\end{equation}
In a driftless theory $V=0$ and $\delta_1\geq 1$ as
$$
D(\lambda^\alpha_0Z_\alpha)=\lambda_1^\alpha Z_\alpha
$$
and equation (\ref{2}) is satisfied with $p=1$. The corresponding
gauge transformation (\ref{gtr}) is the time reparametrization.

We use the decomposition (\ref{decomp}) to construct a basis of
undepressible gauge transformations for the equations of motion
(\ref{pr}). If $\{Z_{\alpha_p}\}_{\alpha_p=1}^{\delta_p}$ is a
basis in the $\lambda$-distribution $\mathcal{Z}^p$,  then the
undepressible gauge transformations read
\begin{equation}\label{gtrz}
    \delta_{\varepsilon}x^i=\sum_{n=0}^{p-1}
    R^i_{_{(p-n-1)}\alpha_p}\stackrel{_{(n)}}{\varepsilon}{}^{\alpha_p}\,,\qquad
    \delta_{\varepsilon}\lambda^\alpha=\sum_{n=0}^{p}U_{_{(p-n)}\alpha_p}^\alpha\stackrel{_{(n)}}{\varepsilon}{}^{\alpha_p}\,,
\end{equation}
where
\begin{equation}\label{R's}
    R_{_{(n)}\alpha_p}^i=\sum_{m=0}^n D^mW^i_{_{(n-m)}\alpha_p}\,,\qquad
    W_{_{(n)}\alpha_p}^i=U_{_{(n)}\alpha_p}^\alpha Z_\alpha^i\,,\qquad
    U_{_{(0)}\alpha_p}^\alpha=\delta_{\alpha_p}^\alpha\,.
\end{equation}

\begin{prop}\label{Prop2}
If $w\in \mathcal{Z}^p$, then the vector fields $w$, $Dw$, $D^2w$,
..., $D^{p-1}w$ are linearly  independent.
\end{prop}
This is true, otherwise it would be $w\in \mathcal{Z}^q$ with
$q\leq p$, because of Proposition \ref{Prop1}.

\begin{cor} The vector fields $R_{^{(n)\alpha_p}}^i$ entering (\ref{gtrz})  are linear
independent, so that $\delta_\varepsilon x^i$ involves all the
successive  derivatives of the gauge parameter up to the order
$k$.
\end{cor}
Let us denote
$$
{\mathcal{Z}}_D=\mathrm{span}_{\mathcal{F}}\big\{D^mu\,|\,\forall
u\in \mathcal{Z}, m\in \mathbb{N}\big\}\,.
$$
By construction,  the $\lambda$-distribution  ${\mathcal{Z}}_D$ is
invariant under the action of $D$, i.e., $D{\mathcal{Z}}_D\subset
{\mathcal{Z}}_D$. It turns out  that the distribution
${\mathcal{Z}}_D$ is actually involutive and coincides with
${\mathcal{Z}}_V$.

 \begin{prop}\label{Prop3} ${\mathcal{Z}}_D={\mathcal{Z}}_V$.
\end{prop}

\begin{cor}
The gauge variations (\ref{gtrz}) of $x$'s are spanned by
$\mathcal{Z}_V$.
\end{cor}

\begin{proof} Since ${\mathcal{Z}}_D\subset {\mathcal{Z}}_V$, it remains to
prove the converse inclusion. Let us  first consider the special
case where $V=0$. For a driftless system
$\mathcal{Z}_V=\bar{\mathcal{Z}}$ and we must  show that
$\mathcal{Z}_D$ contains the Lie closure of  $\mathcal{Z}$.

Increasing, if necessary, the order of the derivatives of the
gauge parameter by changing $\varepsilon \mapsto
\stackrel{_{(k)}}{\varepsilon}$, the gauge transformations
(\ref{gtrz}) can be brought to the form
\begin{equation}\label{GT}
    \delta_\varepsilon x^i = \stackrel{_{(n)}}{\varepsilon}{}^\alpha
    Z_\alpha^i(x)+ \cdots\,,\qquad
    \delta_\varepsilon\lambda^\alpha=\stackrel{_{(n+1)}}{\varepsilon}{}^\alpha
    + \cdots\,,
\end{equation}
where the dots stand for the terms with the lower order
derivatives of the gauge parameter. As the transformations
(\ref{GT}) exhaust \textit{all} the gauge symmetries of the
equations, they have to form an on-shell closed gauge algebra with
respect to the commutator of infinitesimal transformations:
\begin{equation}\label{Com}
    [\delta_{\varepsilon_1},\delta_{\varepsilon_2}]|_{\mathrm{on-shell}}=
    \delta_{\varepsilon_3}\,,
\end{equation}
where
\begin{equation}\label{}
    \varepsilon^\gamma_3=\sum_{n,m}\stackrel{_{(n)}}{\varepsilon}_1{}^{\alpha
    }f_{\alpha n\beta
    m}^\gamma(x,\lambda_k)\stackrel{_{(m)}}{\varepsilon}_2{}^{\beta}\,,
\end{equation}
with $f$'s being the structure functions of the gauge algebra. On
the other hand,
\begin{equation}\label{}
    [\delta_{\varepsilon_1},\delta_{\varepsilon_2}]x^i=\stackrel{(n)}{\varepsilon}{}_1^\alpha\stackrel{(n)}{\varepsilon}{}_2^\beta[Z_\alpha,Z_\beta]^i
    +\cdots\,,
\end{equation}
where the dots stand for the \textit{other} bilinear combinations
of the derivatives
$\stackrel{_{(k)}}{\varepsilon}{}^\alpha_{1,2}$. At every fixed
instant of time, the derivatives  of the parameters
$\stackrel{_{(k)}}{\varepsilon}{}^\alpha_{1,2}$ can take on
arbitrary predetermined values. Then, comparing the right and left
hand sides of (\ref{Com}), we conclude that all commutators
$[Z_\alpha,Z_\beta]$ are given by linear combinations of elements
from ${\mathcal{Z}}_D$. A similar analysis for the successive
commutators of the gauge transformations
$$
[\delta_{\varepsilon_m},[\delta_{\varepsilon_{m-1}},\cdots\,,[\delta_{\varepsilon_2},\delta_{\varepsilon_1}]\cdots
]
$$
shows that all multiple commutators
$$[Z_{\alpha_m},[Z_{\alpha_{m-1}},\cdots\,,[Z_{\alpha_2}, Z_{\alpha_1}]\cdots]$$
are also included into ${\mathcal{Z}}_D$. In other words, we see
that ${\bar{\mathcal{Z}}}\subset {\mathcal{Z}}_D$.

The general case of a non-vanishing drift, $V\neq 0$, can be
formally reduced to the previous one by the following trick. Let
us associate to equations (\ref{pr}) another dynamical system
\begin{equation}\label{gauging}
    \dot x^i =e V^i(x)+ \lambda^\alpha Z_\alpha(x)\,,
\end{equation}
where $e$ is a new undetermined multiplier. The system
(\ref{gauging}) is driftless, and hence reparametrization
invariant. The characteristic distribution of (\ref{gauging}) is
generated by the vector fields $V$ and $Z_\alpha$. The basis of
infinitesimal gauge transformations for this new system can be
chosen in the following way. First of all, the system is invariant
under reparametrizations. The corresponding gauge transformation
reads
\begin{equation}\label{rep}
    \delta_{\epsilon}x^i=\epsilon\dot x^i=\epsilon(eV^i+\lambda^\alpha
    Z_\alpha)\,,\qquad \delta_\epsilon \lambda^\alpha=\dot \epsilon \lambda^\alpha+\epsilon\dot\lambda{}^\alpha\,,
    \qquad \delta_\epsilon e=\dot \epsilon e+\epsilon\dot e\,.
\end{equation}
As is seen the generator of this transformation involves the
vector field $V$ and acts nontrivially on the new multiplier $e$.
Considering $e$ as a fixed function of time, we then define the
gauge transformations of the form \footnote{As before, we have
risen here the order of the gauge transformations to a certain
uniform value $n$.}:
\begin{equation}\label{transf}
    \delta_{\varepsilon} x^i=
    \stackrel{_{(n)}}{\varepsilon}{}^\alpha Z^i_\alpha+\cdots \,,\qquad
    \delta_{\varepsilon}\lambda^\alpha=\stackrel{_{(n+1)}}{\varepsilon}{}^\alpha
    +\cdots\,,\qquad \delta_{\varepsilon}e=0\,,
\end{equation}
The existence of these transformations easily follows from
Proposition \ref{Prop1}. According to the general formulas
(\ref{gtrz}), (\ref{R's}) the expansion coefficients of
$\delta_{\varepsilon}x^i$ in the time derivatives of
$\varepsilon$'s are given by linear combinations of the vector
fields
 $ D^m_eZ_{\alpha}$, where
$$D_e=-\partial-e[V,\;\cdot\;]-\lambda^\alpha[Z_{\alpha},\;\cdot\;]\, .$$
Since the leading terms in the variations $\delta_\epsilon x^i$
and $\delta_{\varepsilon}x^i$ span the entire characteristic
distribution $\mathcal{E}=\mathrm{span}\{V,Z_\alpha\}$, the
transformations (\ref{rep}), (\ref{transf}) exhaust all the gauge
symmetries of (\ref{gauging}).

We claim that the $\varepsilon$-transformations (\ref{transf})
constitute an ideal of the algebra of all gauge transformation (
\ref{rep}), (\ref{transf}), i.e.,
\begin{equation}\label{ideal}
[\delta_{\varepsilon_1},\delta_{\varepsilon_2}]|_{\mathrm{{on-shell}}}=\delta_{\varepsilon_3}\,,\qquad
 [ \delta_{\epsilon},\delta_{\varepsilon_1} ]|_{\mathrm{on-shell}}=\delta_{\varepsilon_2}\,.
\end{equation}
To prove this statement it is enough to apply the commutators in
the l.h.s. of the relations above to $e$. We have
\begin{equation}\label{}
[\delta_{\varepsilon_1},\delta_{\varepsilon_2}]e=0\,,\qquad [
\delta_{\epsilon},\delta_{\varepsilon_1} ]e =0\,.
\end{equation}
Thus the reparametrization transform (\ref{rep}) does not
contribute to the r.h.s. of (\ref{ideal}).

Arguments similar to those we have used in the driftless case
allow one to prove that all the multiple commutators of the vector
fields $V$ and $Z_\alpha$ are contained in $\mathcal{Z}_D$.
Consider, for example, the following commutator of the gauge
transformations:
\begin{equation}\label{com1}
    [\delta_\epsilon,\delta_{\varepsilon}]x^i=
    \epsilon \stackrel{_{(n)}}{\varepsilon}{}^\alpha
    e[V,Z_\alpha]x^i+\cdots\,.
\end{equation}
Because of relations (\ref{ideal}) the right hand side of
(\ref{com1}) must be a gauge transformation of the form
(\ref{transf}). We have
\begin{equation}\label{exp}
 [\delta_\epsilon,\delta_{\varepsilon}]x^i=\delta_{\varepsilon'}x^i\,,\qquad
    \varepsilon'{}^\gamma=\epsilon \stackrel{_{(n)}}{\varepsilon}{}^\alpha
    f^\gamma_{\alpha n}(x,e_s,\lambda_k)+ \cdots\,,
\end{equation}
where the dots stand for the other combinations of the time
derivatives of $\epsilon$ and $\varepsilon^\alpha$. Comparing the
coefficients at $\epsilon \stackrel{_{(n)}}{\varepsilon}$ in
(\ref{com1}) and (\ref{exp}), we see that every commutator
$e[V,Z_\alpha]$ is given by a linear combination of
$D^m_eZ_\alpha$. Setting $e=1$, we conclude that $[V,Z_\alpha]\in
\mathcal{Z}_D$.

In a similar manner, one can see that all the higher iterated
commutators of $Z_\alpha$ and $V$ belong to $\mathcal{Z}_D$.

\end{proof}

By construction, the structure functions (\ref{R's}) are
meromorphic functions of  $x^i$ and
$\stackrel{_{(s)}}{\lambda}{}^{\alpha}$. So, the
 trajectories $(x^i(t),\lambda^\alpha(t))$ can exist such that the
 transformations (\ref{gtrz}) are ill defined. It is easy to see,
however, that the functions (\ref{R's}) can be made real analytic
by an appropriate change of the gauge parameters:
\begin{equation}\label{}
    \varepsilon^\alpha \rightarrow
    \tilde{{\varepsilon}}{}^\alpha=U^\alpha_\beta\varepsilon^\beta\,,
\end{equation}
with $\det (U^\alpha_\beta)$ being a nonzero element of
$\mathcal{F}$. The real analytic transformations are well defined
on all the trajectories.

\vspace{3mm} \noindent {\textit{{Example}}}. Consider the
driftless system
\begin{equation}\label{example}
    \dot x^i=\lambda^\alpha Z_\alpha^i
\end{equation}
associated to the following characteristic distribution in
$\mathbb{R}^{10}$:
\begin{equation}\label{V}
    \begin{array}{l}
      \displaystyle Z_1=\frac{\partial}{\partial x^1}+x^2\frac{\partial}{\partial x^8}+x^3\frac{\partial}{\partial x^9}+x^4\frac{\partial}{\partial
       x^{10}}\,,\\[4mm]
    \displaystyle  Z_2=\frac{\partial}{\partial x^2}+x^3\frac{\partial}{\partial x^6}+x^4\frac{\partial}{\partial x^7}\,,\\[4mm]
    \displaystyle  Z_3=\frac{\partial}{\partial x^3}+x^4\frac{\partial}{\partial x^5}\,,\\[4mm]
    \displaystyle  Z_4=\frac{\partial}{\partial x^4}\,.\\
    \end{array}
\end{equation}

\noindent A straightforward computation shows  that the vector
fields $\{Z_\alpha, [Z_\beta,Z_\gamma]\}$ span the entire tangent
space of $\mathbb{R}^{10}$ and
 $$
[[Z_\alpha,Z_\beta],Z_\gamma]=0\,.
 $$
Introduce the vector field $X_n=\lambda_n^\alpha Z_\alpha$ and the
differential $D=-\partial- [X_0,\;\cdot\;]$. One can readily check
that
\begin{equation}\label{DV}
    \begin{array}{l}
      DX_0+X_1=0 \,,\\[3mm]
      D^2X_1+ 2DX_2+X_3=0\,,\\[3mm]
      D^3X_2+3D^2X_3+3DX_4+X_5=0\,, \\[3mm]
      D^4X_3+4D^3X_4+6D^2X_5+4DX_6+X_7=0\,. \\
    \end{array}
\end{equation}

\noindent These equalities are the particular cases of the general
identity
\begin{equation}\label{id}
    \sum_{k=0}^m C^k_mD^{m-k}X_{m+k-1}=0\,.
\end{equation}

 The vector fields $\{X_0,X_1,X_2,X_3\}$ form another basis in
the $\lambda$-distribution
$\mathrm{span}_\mathcal{F}\{Z_\alpha\}$. Comparing relations
(\ref{DV}) with (\ref{gtrz}), (\ref{R's}) and (\ref{compcond}), we
get the following gauge transformations for the differential
equations (\ref{example}):
$$
    \begin{array}{rcl}
      \delta_\varepsilon x & = & \varepsilon_1 X_0 \\[3mm]
       & + & \dot \varepsilon _2 X_1 + \varepsilon_2 (D X_1+2 X_2)\\[3mm]
       & + & \ddot \varepsilon_3 X_2 +\dot \varepsilon_3 (DX_2+3X_3) + \varepsilon_3
       (D^2 X_2+3DX_3+3X_4)\\[3mm]
       & + &  \dddot \varepsilon_4 X_3 +\ddot\varepsilon_4(DX_3+4X_4)+\dot \varepsilon_4(D^2X_3+4DX_4+6X_5)+\varepsilon_4
       (D^3X_3+4D^2X_4+6DX_6+4 X_6)\,,\\[5mm]
  \delta_\varepsilon \lambda& = & \dot\varepsilon_1\lambda_0 +\varepsilon_1
\lambda_1 \\[3mm]
   & + & \ddot \varepsilon_2\lambda_1 +2\dot\varepsilon_2\lambda_2 + \varepsilon_2\lambda_3
   \\[3mm]
   & + &  \dddot \varepsilon_3\lambda_2 +3\ddot\varepsilon_3\lambda_3 + 3\dot\varepsilon_3\lambda_4+\varepsilon_3\lambda_5
   \\[3mm]
   & + & \ddddot \varepsilon_4\lambda_3 +4\dddot\varepsilon_4\lambda_4 + 6\ddot\varepsilon_4\lambda_5+4\dot\varepsilon_4\lambda_6+\varepsilon_4\lambda_7 \,.\\
\end{array}
$$
Here $\lambda_0^\alpha=\lambda^\alpha$ and
$\lambda_{k+1}^\alpha={\dot\lambda}{}_k^\alpha$.

\vspace{3mm}

Consider now the general system (\ref{pr}) including both the
differential equations and the constraints. It follows from
Proposition \ref{Prop3} that all the gauge symmetries of the
differential equations alone are also the symmetries of the
constraints. Indeed, the distribution $\mathcal{E}$ is tangent to
the complete constraint surface $\bar{\Sigma}$, hence (by the
Frobenius theorem) the closure $\bar{\mathcal{E}}$ is also tangent
to $\bar \Sigma$. This means that not only the vector fields
$Z_\alpha$ and $V$ respect the constraint surface, but also all
their iterated commutators do the same:
\begin{equation}\label{}
    XT_a={F}(X){}_a^bT_b \qquad \forall X\in \bar{\mathcal{E}}\,.
\end{equation}
On the other hand, according to Corollary 3 from Proposition
\ref{Prop3}, the gauge variation $\delta_\varepsilon x^i$ is to be
spanned by the vectors from the distribution $\mathcal{Z}_V\subset
\bar{\mathcal{E}}$. This immediately gives $\delta_\varepsilon
T_a|_{\bar\Sigma}=0$.

An important point to stress is that the number of
\textit{nontrivial} gauge transformations can decrease when the
constraints are taken into account. The matter is that some of the
structure functions\footnote{As it has been already noticed, all
these functions can be chosen to be real analytic.} $R_n$ and
$U_n$, which are involved into the gauge transformation
(\ref{gtr}), can be trivial (\ref{triv}) with regard to the
constraints. In other words, some linear combinations of the gauge
transformations can vanish identically on $\bar \Sigma$. Such
transformations are also called trivial. All the gauge symmetry
transformations are to be considered modulo trivial ones.

To formulate a systematic algorithm for constructing a basis of
nontrivial and linearly independent gauge transformations in the
presence of constraints we need some  algebraic background. Given
a constrained dynamical system in the complete normal form
(\ref{pr}), denote by $\mathcal{R}$ the ring of analytical
functions of $x^i$ and of a finite number of variables
$\lambda^\alpha_k$.  Let $\mathcal{I}\subset \mathcal{R}$ denote
the principle ideal generated by the regular constraints
$\{T_a\}$. For simplicity sake, assume  that the ideal
$\mathcal{I}$ is simple\footnote{An ideal $\mathcal{I}\subset
\mathcal{R}$ is said to be simple if $ab\in \mathcal{I}$ implies
either $a\in \mathcal{I}$ or $b\in \mathcal{I}$. In the case where
$\mathcal{R}$  is the ring of analytical functions,   we have the
following criterion of simplicity: If the constraint surface
$\Sigma\subset \mathbb{R}^n$ associated to a set of regular
constraints $T_a=0$ is \textit{connected}, then the principal
ideal $\mathcal{I}=\langle T_a\rangle$ is simple.}. Then the
quotient $\mathcal{R}/\mathcal{I}$ is an integrality domain and we
can form the field of fractions $\mathcal{F}_{\mathcal{I}}
=\mathrm{Fr}(\mathcal{R}/\mathcal{I})$. The field
$\mathcal{F}_{\mathcal{I}}$ is a natural substitution for the
field of meromorphic functions $\mathcal{F}$ in the presence of
constraints. As a practical matter, it is more convenient to use
the following equivalent definition of
$\mathcal{F}_{\mathcal{I}}$. A meromorphic function  $f\in
\mathcal{F}$ is said to be regular if it admits a representation
$f=a/b$, where $a,b \in \mathcal{R}$ and  $b\notin \mathcal{I}$.
The ideal $\mathcal{I}$ being simple, all the regular meromorphic
functions constitute a ring $\mathcal{R}_{\mathcal{I}}\subset
\mathcal{F}$. It is easily  seen that $\mathcal{I}$ is the maximal
proper  ideal of $\mathcal{R}_{\mathcal{I}}$ and
$\mathcal{F}_{\mathcal{I}}=\mathcal{R}_{\mathcal{I}}/\mathcal{I}$.
Thus $\mathcal{F}_{\mathcal{I}}$ is  just a subquotient of
$\mathcal{F}$ and in all practical calculations we can replace the
elements of the field $\mathcal{F}_{\mathcal{I}}$ by their regular
representatives in $\mathcal{F}$.

Now define the $n$-dimensional vector space
$\mathcal{W}_{\mathcal{I}}$ over $\mathcal{F}_{\mathcal{I}}$ as
\begin{equation}\label{}
\mathcal{W}_{\mathcal{I}}=\mathrm{span}_{\mathcal{F}_\mathcal{I}}\left(\frac{\partial}{\partial
x^1}\,,\cdots \,,\frac{\partial}{\partial x^n}\right)\,.
\end{equation}
Again, we can view $\mathcal{W}_\mathcal{I}$ as a subquotient of
$\mathcal{W}$ and represent the vectors of
$\mathcal{W}_{\mathcal{I}}$ by regular elements of the
$\lambda$-distribution $\mathcal{W}$, i.e., those vectors of
$\mathcal{W}$ whose components are regular meromorphic functions
of $\mathcal{F}$.

Since  the vector fields $V$ and $Z_\alpha$ entering the
definition of our dynamical system (\ref{pr}) are assumed to be
regular and linearly independent modulo constraints, we can define
the $m$-dimensional  subspace $\mathcal{Z}_{\mathcal{I}}\subset
\mathcal{W}_{\mathcal{I}}$ as
\begin{equation}\label{}
    \mathcal{Z}_{\mathcal{I}}=\mathrm{span}_{\mathcal{F}_\mathcal{I}}\left(Z_1,\ldots,Z_m\right)\,.
\end{equation}
Using the fact that the  action of the vector fields $V$ and
$Z_\alpha$ preserves the constraints $T_a$, one can easily see
that the $\mathbb{R}$-linear operator $D: \mathcal{W}\rightarrow
\mathcal{W}$ defined by (\ref{D}) induces an $\mathbb{R}$-linear
operator in the subquotient $\mathcal{W}_{\mathcal{I}}$. We will
denote the latter operator by the same symbol $D$. Similar to the
case of unconstrained system we can filter the space
$\mathcal{Z}_I$ by the finite
 sequence  of subspaces
\begin{equation}\label{}
    \mathcal{Z}_\mathcal{I}\supset
    \mathcal{Z}^{_{(N)}}_{\mathcal{I}}\supset
    \mathcal{Z}^{_{(N-1)}}_{\mathcal{I}}\supset\cdots\supset
    \mathcal{Z}^{_{(1)}}_{\mathcal{I}}\supset 0\,,
\end{equation}
where
\begin{equation}\label{}
    Z^{_{(p)}}_{\mathcal{I}}=\left\{u\in \mathcal{Z}_{\mathcal{I}}\;\big |\; D^pu\in \mathcal{Z}_\mathcal{I}\cup D\mathcal{Z}_\mathcal{I}\cup\cdots\cup D^{p-1}
    \mathcal{Z}_\mathcal{I}\right\}\,.
\end{equation}
Using this filtration and the Euclidean metric in $\mathbb{R}^n$,
we can then split $\mathcal{Z}_\mathcal{I}$ in the direct sum of
subspaces
\begin{equation}\label{}
    \mathcal{Z}_\mathcal{I}=\mathcal{Z}^{1}_\mathcal{I}\oplus
    \mathcal{Z}^{2}_\mathcal{I}\oplus\cdots\oplus
    \mathcal{Z}^{N}_\mathcal{I}\,,\qquad
    \mathcal{Z}^p_\mathcal{I}\simeq
    \mathcal{Z}^{_{(p)}}_\mathcal{I}/\mathcal{Z}_\mathcal{I}^{_{(p-1)}}\,.
\end{equation}
Now a straightforward analog of Proposition \ref{Prop1} for the
differential $\mathcal{F}_\mathcal{I}$-module
$\mathcal{W}_\mathcal{I}$ states that for any $w_0\in
\mathcal{Z}^p_\mathcal{I}$ there exists a sequence of elements
$w_1$,...,$w_p\in \mathcal{Z}_\mathcal{I}$ such that
\begin{equation}\label{}
    D^pw_0+D^{p-1}w_1+\cdots +Dw_{p-1}+w_p=0\,.
\end{equation}
Let $W_0$,..., $W_p\in \mathcal{W}$  be  regular representatives
of the elements $w_0,...,w_p$.  Then, according to the general
formulae (\ref{gtrz}), the gauge transformations read:
\begin{equation}\label{gtrr}
    \delta_\varepsilon x^i=\sum_{n=0}^{p-1}\sum_{m=0}^{p-n-1}\stackrel{_{(n)}}{\varepsilon} D^mW^i_{{p-n-m-1}}
    \,,\qquad
    \delta_\varepsilon
    \lambda^\alpha= \sum_{n=1}^{p}\stackrel{_{(n)}}{\varepsilon}
    U^\alpha_{{p-n}}\,,
\end{equation}
where $W_n^i=U^\alpha_{{n}}Z^i_\alpha$. Let us choose a basis
$\{Z_{\alpha_p}\}$ in every subspace $\mathcal{Z}^p_\mathcal{I}$.
To any basis element we can  associate a gauge transformation of
the form (\ref{gtrr}) and this yields a complete basis of
undepressible gauge transformations in the presence of
constraints.

We conclude this section by some remarks concerning interpretation
of $\mathcal{Z}_V$ in control theory and the theory of gauge
systems. Observe that the restriction
$\mathcal{Z}_V|_{\bar{\Sigma}}$, being a completely integrable
distribution of $\bar\Sigma$, endows the constraint surface with
the structure of foliation $\mathcal{F}(\bar \Sigma)$. In the
context of gauge systems, the leaves of this foliation are know as
\textit{gauge orbits}.  Two points of the constraint surface
$\bar\Sigma$ are considered equivalent if they belong to the same
gauge orbit. Equivalent points define the same physical state, so
that the space of all physical states of the gauge system is
identified with the space of leaves
$\bar{\Sigma}/\mathcal{F}(\bar\Sigma)$. On the other hand, one of
the basic concepts of control theory is the notion of
\textit{attainable set}. By definition, a point $y\in \bar \Sigma$
belongs to the attainable set of a point $x\in \bar\Sigma$ if one
can join $y$ to $x$ by an integral curve of (\ref{pr}) with some
fixed functions $\lambda^\alpha(t)$. In the case where
$\mathcal{Z}_V$ and $\bar \Sigma$ are real anaclitic and $V\in
\mathcal{Z}_V$, the so-called \textit{orbit theorem } \cite{SA}
ensures that the attainable set of a point $x\in \bar \Sigma$
coincides with the gauge orbit passing through $x$. This
coincidence is not particularly surprising, since the gauge
transformations, involving arbitrary functions of time as
parameters, allow the (control) functions $\lambda^{\alpha}(t)$ to
take on arbitrary values at each given instant of time. So, the
doctrines of control and gauge theories are in a sense
complementary: the controllable part of dynamics is non-physical,
while the physical part is uncontrollable.

\section{Physical observables and an involutive normal form}

By definition, $t$-local values associated to a  dynamical system
in the complete form (\ref{pr}) are functions of the phase-space
coordinates $x^i$, undetermined multipliers $\lambda^\alpha$ and
their derivatives up to some finite order. Note that the time
derivatives of $x^i$ can always be excluded with the help of the
equations of motion (\ref{pr}). Therefore, without loss of
generality, we can identify the space of $t$-local values with the
space of real analytical functions $\mathcal{R}$. Among these
functions, there are trivial ones (\ref{triv}) that vanish on the
complete constraint surface $\bar{\Sigma}$. Two $t$-local values
$O_1$ and $O_2$ are said to be \textit{equivalent} if their
difference is a trivial function. In view of the regularity
assumptions, the last condition amounts to
\begin{equation}\label{sim}
O_1\sim O_2\qquad \Leftrightarrow\qquad  O_1-O_2 = F^aT_a\,.
\end{equation}
In the previous section, the space of equivalence classes was
identified with the quotient $\mathcal{R}/\mathcal{I}$, where
$\mathcal{I}$ is the ideal generated by the constraints.

The infinitesimal gauge transformation (\ref{gtrr}) maps any
solution of (\ref{pr}) to another one. Given an initial time
moment $t_0$, the gauge parameters $\varepsilon$ can be chosen
vanishing together with all their derivatives involved in the
gauge transform. So, the initial data remain the same, while the
solutions are different. The physical values should evolve in a
casual way, i.e., they should take the same value on every
solution originating from a given initial state. This implies that
the physical values are to be on-shell invariants of the gauge
transformations. As is seen from (\ref{gtrr}), the gauge variation
of $\lambda$'s starts with the highest time derivative of
$\varepsilon$'s and this derivative  does not contribute to the
gauge transformation of $x$'s. This suggests  that the gauge
invariant $t$-local values can depend on $\lambda$'s only through
the trivial contributions. In other words, each $t$-local physical
value can be represented by a function of $x^i$. In view of
Proposition \ref{Prop3} the gauge transformations for the
functions of $x^i$ are generated by the gauge distribution
${\mathcal Z}_V$ so that the subspace of trivial values is
automatically gauge invariant. Thus we are lead to the following
definition: A \textit{physical observable} of a dynamical system
brought to the complete normal form (\ref{pr})  is the equivalence
class of a phase-space function $O(x)$ that is on-shell invariant
under the action of the gauge distribution ${\mathcal Z}_V$, i.e.,
\begin{equation}\label{Observable}
ZO|_{\bar{\Sigma}} = 0\qquad \forall Z\in \mathcal{Z}_V\,.
\end{equation}

For the physical  observables, the equations of motion (\ref{pr})
reduce to the form
\begin{equation}\label{dO}
    \dot{O} = V O \,,
\end{equation}
where $V$ is the complete drift. The undetermined multipliers
$\lambda^\alpha$ drop out of these equations as the physical
observables are invariant under the action of the characteristic
distribution $\mathcal{Z}\subset\mathcal{Z}_V$. Notice that the
time derivative of an observable is again an observable because
$[V,{\mathcal Z}_V]\subset \mathcal{Z}_V$ and the complete
constraint surface $\bar{\Sigma}$ is invariant under the action of
$V$ and ${\mathcal Z}_V$. With initial data specified, the unique
existence of the solution $O(t)$, $t\geq t_0$, to equation
(\ref{dO}) follows from two facts: (i) the undetermined
multipliers are not contained in the equations and (ii) the
equivalence class of $VO$ is gauge invariant.  So, the right hand
side of (\ref{dO}) is the same for any solution $x^i(t)$ evolving
from a given initial state $x^i(t_0)=x^i_0$. All that confirms
ones again that the definition of the physical observables
provides them casual evolution.

As is seen, the following data are only needed to define the
physical observables  and their time evolution: the phase space
$U$, the complete constraint surface $\bar{\Sigma}$, the complete
drift $V$, and the gauge distribution ${\mathcal Z}_V$. The
quadruple $(U,\bar{\Sigma},V,\mathcal{Z}_V)$ can always be
unambiguously derived from equations of motion in their primary
normal form (\ref{lambda}), (\ref{Constr1}) following  the
algorithm of the previous sections. The converse is not true:
different primary normal forms can result in  the same quadruple
$(U, \bar \Sigma, V, {\mathcal Z}_V)$. These differences can be
much larger than just the equivalence relations (\ref{change}),
(\ref{aut}) for the primary normal form. In particular, the output
can be the same, even though the algorithm has been applied to
dynamical systems with characteristic distributions $\mathcal Z$
and primary constraint surfaces $\Sigma$ of different dimensions.
With the dynamics in mind of physical observables, it seems
reasonable to consider two dynamical systems as being equivalent
whenever they have the same complete constraint surfaces, complete
drifts and coinciding gauge distributions\footnote{In control
theory, two affine-control systems are called \textit{feedback
equivalent} whenever every solution $x(t)$ of equations
(\ref{lambda}) with every given control $\lambda (t)$ of one
system coincides with a certain solution of another one, possibly
with a different control. The feedback equivalence is much more
restrictive notion than the definition above: For feedback
equivalent systems all values $f(x)$, both controllable and
uncontrollable, have the same time evolutions, while for
physically equivalent systems, unobservable (i.e., controllable)
values may evolve differently.}.

Among the systems that are physically equivalent to
(\ref{lambda}), (\ref{Constr1}) there is a special one whose
characteristic distribution, primary constraint set, and primary
drift coincide, correspondingly, with the gauge distribution,
complete constraint set, and the complete drift of the original
system:
\begin{equation}\label{involutive}
\dot{x}^i= \bar{V}{}^i(x) + \lambda^{\bar \alpha} Z_{\bar
\alpha}^i(x) \, , \qquad {T}_{\bar a}(x)=0 \, .
\end{equation}
By construction, ${\mathcal
Z}_{\bar{V}}=\mathrm{span}\{Z_{\bar\alpha}\}$ is tangent to the
complete constraint surface $\bar\Sigma=\{x\in U| T_{\bar
a}(x)=0\}$. This means no compatibility conditions can arise from
(\ref{involutive}). Also,
 ${\mathcal Z}_V$ is involutive and invariant with respect to the complete drift $\bar V$.
Therefore, the gauge distribution for (\ref{involutive}) coincides
with the characteristic distribution. We call equations
(\ref{involutive}) the \textit{involutive normal form} of local
dynamics. The system  (\ref{involutive}) describes the observables
defined by the same conditions (\ref{Observable}) and having the
same evolution law (\ref{dO}) as the physical observables
associated to the original equations (\ref{lambda}),
(\ref{Constr1}). As the characteristic distribution is involutive,
the gauge transformations for (\ref{involutive}) take  quite a
simple form (\ref{gt0}).

Let us comment on the involutive normal form (\ref{involutive})
for a system whose primary normal form is the constrained
Hamiltonian dynamics \cite{Dirac}, \cite{HT}. If equations
(\ref{lambda}), (\ref{Constr1}) follow from the least action
principle (\ref{SDirac}), then the primary characteristic
distribution $\mathcal Z$ is spanned by the Hamiltonian vector
fields for the primary constraints and the drift is the
Hamiltonian vector field for the primary Hamiltonian. As it has
been already explained in Section 3, the transverse constraints
$T^{^\bot}$ correspond to the second-class constraints, the
transverse distribution ${\mathcal Z}_\bot$ is spanned by the
Hamiltonian vector fields for the primary second-class
constraints, and $\mathcal{Z}_\|$ is generated  by the Hamiltonian
vector field associated  to the primary first-class constraints.
Suppose we have no transverse constraints. Then the characteristic
distribution is tangent to the complete constraint surface. From
the viewpoint of Dirac's classification, this is the case of a
pure first-class system. Whenever the Dirac conjecture is
true\footnote{It is not always true, see \cite{HT} for
counterexamples.}, the Hamiltonian vector fields for \textit{all}
the first-class constraints generate the gauge transformations and
the involutive form of dynamics (\ref{involutive}) is again
variational. The corresponding action
\begin{equation}\label{TDirac}
    S[x,\lambda]=\int\left(\rho_i(x)\dot{x}{}^i - H_{\mathrm{tot}}(x,\lambda)\right)dt \, ,
\end{equation}
involves the total Hamiltonian $H_{\mathrm{tot}}=H +\lambda^{\bar
a}{T}_{\bar a}$ given by the sum of the original Hamiltonian
(\ref{Dirac}) and the linear combination of all the first-class
constraints (both primary and secondary) with independent Lagrange
multipliers. In this case, the gauge distribution is spanned by
the Hamiltonian vector fields $X_{\bar a}=\{T_{\bar
a},\;\cdot\,\}$ for the first-class constraints.

Whether the Dirac conjecture is true or not, the algorithm of
Section 4 shows that the gauge distribution is included into
$\mathrm{span} \{ X_{\bar a}\}$ provided that the primary normal
form was variational. This does not mean, however, that
\textit{every} Hamiltonian vector field $X_{\bar a}$ should belong
to the gauge distribution. In case $\mathcal{Z}_V\neq
\mathrm{span}\{X_{\bar a}\}$,  the property conjectured by Dirac
does not hold and the involutive form of the constrained
Hamiltonian dynamics is \textit{not variational} anymore. Be it as
it may, the algorithm proposed in this paper will automatically
separate the constraints whose Hamiltonian vector fields
contribute to the gauge transformations from those which
Hamiltonian vector fields do not. This allows us to systematically
identify all the true gauge symmetries for any regular,
constrained Hamiltonian system, even though the system does not
satisfy the Dirac conjecture.

Concluding this section, let us briefly discuss the  issue of
equipping  the involutive dynamics (\ref{involutive}) with a
certain Hamiltonian structure when the original  equations of
motion are not variational. The basic idea is that only the
algebra of physical observables (\ref{Observable}), not the
algebra of all $t$-local values, should be equipped with the
Poisson bracket. In other words, it is sufficient to require the
Jacobi identity to hold only when the Poisson  bracket is applied
to a pair of observables. Besides, the Poisson bracket is to be
compatible with the time evolution (\ref{involutive}) in the sense
that the time derivative should differentiate the bracket of two
observables (not arbitrary functions) by the Leibnitz rule. In
\cite{LS1}, we introduced a notion of \textit{ weak Hamiltonian
structure}, which satisfies all the above properties. A similar
construction was also studied in \cite{CaFe}.

Formally, a weak Hamiltonian structure associated to an involutive
dynamical system is defined by the quadruple $(T,Z,V,P)$, where
$T=\{T_a\}$ is a set of constraints, $Z=\{Z_\alpha\}$ are the
generators of a characteristic(=gauge) distribution, $V$ is a
drift, and $P=P^{ij}\partial_i\wedge\partial_j$ is a weak Poisson
bivector. So all the objects are polyvector fields of degree 0, 1,
and 2. In terms of the Schouten bracket of polyvector fields the
defining relations for a weak Hamiltonian structure read
\begin{equation}\label{invol}
  [Z_\alpha,T_a]= A_{\alpha a}^b T_b\,,\qquad [Z_\alpha ,  Z_\beta] = B_{\alpha\beta}^\gamma Z_\gamma +
    T_a C^a_{\alpha\beta}\,,
\end{equation}
\begin{equation}\label{invol'}
  [T_a,V]= D_{a}^b T_b\,,\qquad [Z_\alpha ,  V] = E_{\alpha}^\beta Z_\beta +
    T_a F^a_{\alpha}\,,
\end{equation}
\begin{equation}\label{weakJ}
     [P,P]=G^\alpha\wedge Z_\alpha -T^aH_a\,, \qquad [V,P]=I^\alpha\wedge Z_\alpha + T^a J_a\,,
\end{equation}
\begin{equation}\label{weakP}
    [T_a,P]=K_a^\alpha \wedge Z_\alpha -T_b L^b_a\,,\qquad [Z_\alpha, P]=M_\alpha^\beta \wedge Z_\beta-T_aN^a_\alpha\,,
\end{equation}
where $A,B,C,...,N$ are some polyvector fields. Relations
(\ref{invol}) and (\ref{invol'}) express the fact of involutivity
of the dynamical system. The first relation in (\ref{weakJ})
identifies $P$ as a weak Poisson bivector, i.e., $P$ satisfies the
Jacobi identity  modulo constraints and gauge symmetry generators.
Then the second relation in (\ref{weakJ}) tells us that the
evolution  (\ref{dO}) generated by  $V$ preserves the Poisson
algebra of physical observables (i.e., $V$ is a weakly Poisson
vector field).   Relations (\ref{weakP}) mean that the Hamiltonian
vector fields for the constraints are spanned  by the gauge
generators modulo trivial terms and that the generators are weakly
Poisson vector fields. As a result the trivial functions
constitute the Poisson ideal $\mathcal{I}\subset \mathcal{R}$ that
makes possible to speak about the Poisson algebra of physical
observables $\mathcal{R}/{\mathcal{I}}$.

As it was shown in \cite{LS1}, the weak Hamiltonian structure
(\ref{invol}-\ref{weakP}) admits a nice BRST imbedding that
generalize the usual BFV-BRST formalism for the Hamiltonian
systems with first-class constraints. Starting with this BRST
embedding, it is possible to construct a fully consistent
deformation quantization of a weak Hamiltonian system without any
reference to variational principles. The output is a weakly
associative $\ast$-product inducing an associative quantum
multiplication in the space of physical observables identified
with a certain BRST cohomology. The construction essentially
relies on the superextension of the formality theorem \cite{CaFe}
and may be thought of as a generalization of the Kontsevich
deformation quantization to the case of non-Hamiltonian gauge
theories.

The ``odd counterpart''  of the weak Hamiltonian structure is
known as a \textit{Lagrange structure} \cite{KLS}. The existence
of the latter  structure is much less restrictive for the
equations of motion than the requirement to be variational.
Whereas the weak Hamiltonian structure is aimed at the
construction of $\ast$-product, the Lagrange structure allows one
to perform the path-integral quantization of a (non-)variational
dynamical system. For the variational dynamics, this quantization
is shown to reduce to the standard BV quantization \cite{KLS},
\cite{LS2}. The analysis of classical gauge symmetries performed
in Sections 4, 5 may be viewed as a pre-requisite for extending
the local BRST cohomology techniques \cite{BBH} from the
Lagrangian to non-Lagrangian gauge theories \cite{KLS},
\cite{LS2}. Notice, however, that to make the technique explicitly
covariant in field theoretical context, our analysis should
address the issues of locality in multidimensional space. These
issues are beyond the scope of this work.

\section{Conclusion}

In this paper we propose an algorithm of bringing the general
local dynamics to certain normal forms, which allow us to identify
all gauge symmetries. These normal forms can serve as starting
point for the BRST imbedding and deformation quantization of not
necessarily variational dynamics. Let us briefly summarize the
essentials of the proposed algorithm.

The algorithm becomes applicable after imposing certain regularity
conditions on the local equations of motion. Then, we see that the
general regular equations can be brought to the primary normal
form that includes the differential equations (\ref{lambda}) with
undetermined multipliers and the phase-space constraints
(\ref{Constr1}). In this form, the dynamics are defined by the
three ingredients: the primary constraints, primary characteristic
distribution, and primary drift. In the case of variational
dynamics, this corresponds to Hamiltonian equations subject to
primary constraints. The main problem we address in this paper is
finding all the gauge symmetry transformations (\ref{gt}) for the
equations of motion (\ref{lambda}), (\ref{Constr1}). To this end,
we first consider the differential consequences of the equations
of motion. A basic consistency requirement is the conservation of
the constraints in time.  This results  in a multi-step procedure
of iterating secondary constraints and determining a part of the
multipliers. After terminating the procedure, we are left with
equations (\ref{pr}) - the complete normal form of local dynamics
- which assume no further restrictions on the phase-space
coordinates and/or undetermined multipliers.

Having the dynamical system brought to the complete normal form,
one can go over to finding its gauge symmetries.  We find that the
gauge transformations  (\ref{gtrr}) are generated by the gauge
distribution ${\mathcal Z}_V$. The gauge distribution is the Lie
closure of the tangential characteristic distribution ${\mathcal
Z}$ supplemented by all its iterated commutators with the drift
vector field (\ref{ZViter}). So, the gauge symmetries are
explicitly identified in the same terms that define the original
system (\ref{lambda}), (\ref{Constr1}).

Then we turn to the notion of physical observables. These are
understood as gauge invariants of the dynamics. As the gauge
symmetry is generated by the gauge distribution ${\mathcal Z}_V$,
the observables are defined as equivalence classes (\ref{sim}) of
phase-space functions annihilated by $\mathcal{Z}_V$. This
definition (\ref{Observable}) ensures that the time evolution of
observables (\ref{dO}) is casual. It also suggests to consider two
dynamical systems as equivalent to each other whenever they have
coinciding gauge distributions, complete constraint sets and
complete drifts. As a result, the physical observables and their
time evolutions are the same for all the equivalent systems. In
every equivalence class of the dynamical systems, there is a
special representative whose characteristic distribution is
involutive, tangent to the constraint surface, and is preserved by
the drift. In this form, called the involutive normal, the gauge
transformations take the most simple form (\ref{gt0}). It is the
involutive normal form that serves as a starting point for the
deformation quantization based on the concept of weak Poisson
structure (\ref{invol}-\ref{weakP}).

\vspace{3mm}

\textbf{Acknowledgements.} We are thankful to E.M.~Gorbatenko for
illuminating discussions on regularity of analytic varieties. This
work is partially supported by the RFBR grant no 06-02-17352-a and
by the grant from Russian Federation President Programme of
Support for Leading Scientific Schools no 871.2008.02. SLL is
partially supported by the RFBR grant 08-01-00737-a.


\begin{thebibliography}{33}



\bibitem{Dirac} P.A.M. Dirac, \textit{Lectures on Quantum
Mechanics}, Yeshiva University, New York: Academic Press, 1967.

\bibitem{HT}  M. Henneaux and C. Teitelboim, \textit{Quantization of Gauge
Systems}, Princeton University Press, 1992.



\bibitem{CMP} G. Conte, C.H. Moog, A.M. Perdon,\textit{ Nonlinear control
systems: an algebraic setting}, Springer-Verlag London Limitted,
1999.

\bibitem{SA} A.A. Agrachev and Yu.L. Sachkov, \textit{Control theory
from the geometric viewpoint}, Springer-Verlag Berlin Heidelberg,
2004.


\bibitem{LS1} S.L. Lyakhovich and A.A. Sharapov,
\textit{BRST theory without Hamiltonian and Lagrangian}, JHEP
\textbf{03}(2005)011.


\bibitem{CaFe}  A.S. Cattaneo and  G. Felder,
\textit{Relative Formality Theorem and Quantisation of Coisotropic
Submanifolds}, Adv. Math. \textbf{208} (2007) 521-548.


\bibitem{KLS} P.O.Kazinski, S.L.Lyakhovich and A.A.Sharapov,
\textit{Lagrange Structure and Quantization}, JHEP
\textbf{07}(2005)076.

\bibitem{LS2} S.L. Lyakhovich and A.A. Sharapov,
\textit{Schwinger-Dyson equation for non-Lagrangian field theory},
JHEP \textbf{02}(2006)007.


\bibitem{GP} X. Gracia and J.M. Pons, \textit{Symmetries and infinitesimal symmetries of singular differential
   equations},  J. Phys. \textbf{A35} (2002) 5059.


\bibitem{BBH} G.~Barnich, F.~Brandt, M.~Henneaux, \textit{Local BRST cohomology in gauge
theories}, Phys. Rept. \textbf{338} (2000) 439-569.




\end{thebibliography}
\end{document}